\documentclass[10pt]{article}
\usepackage[nofonts]{dgleich-common}
\pdfoutput=1
\usepackage[T1]{fontenc}
\RequirePackage[scaled=0.8]{helvet}%
\RequirePackage[scaled=0.75]{beramono}%
\def\lstbasicfont{\fontfamily{fvm}\selectfont}
\makeatletter  
\setboolean{@dgleich@sanssmallcaps}{true}
\newcommand*{\sublabel}{%
\let\subcaption@ORI@label=\@tufte@orig@label%
\subcaption@label}

\makeatother

\include{preamble}

\usenatbib
\usehyperref

\usepackage{subcaption}
\usepackage{multirow}
\usepackage{multicol}
\usepackage{paralist}
\usepackage{algorithm}
\usepackage{algorithmic}
\usepackage{bbm}
\newcommand{\smallparagraph}{\paragraph}

\begin{document}

\marginnote[43\baselineskip]{%
\noindent \emph{Yufan Huang \& David F. Gleich} \\
\noindent Purdue University $\cdot$ Computer Science 
\texttt{\{huan1754,dgleich\}@purdue.edu}

\bigskip 
\noindent \emph{C. Seshadhri} \\
\noindent University of California, Santa Cruz $\cdot$ Computer Science \\ 
\noindent \texttt{sesh@ucsc.edu}
\bigskip 

\noindent Huang and Gleich's work was funded in part by 
NSF CCF-1909528, IIS-2007481, 
DOE DE-SC0023162, and the IARPA Agile 
program. \\
\noindent Seshadhri's was funded by NSF DMS-2023495, CCF-1839317, and CCF-1908384.
}

\title{Theoretical bounds on the network community profile from low-rank semi-definite programming}

\author{Yufan Huang $\cdot$ C. Seshadhri $\cdot$ David F. Gleich}

\maketitle

\begin{abstract}
We study a new connection between a technical measure called $\mu$-conductance that arises in the study of Markov chains for sampling convex bodies and the network community profile that characterizes size-resolved properties of clusters and communities in social and information networks. The idea of $\mu$-conductance is similar to the traditional graph conductance, but disregards sets with small volume. We derive a sequence of optimization problems including a low-rank semi-definite program from which we can derive a lower bound on the optimal $\mu$-conductance value. These ideas give the first theoretically sound bound on the behavior of the network community profile for a wide range of cluster sizes. The algorithm scales up to graphs with hundreds of thousands of nodes and we demonstrate how our framework validates the predicted structures of real-world graphs.  
\end{abstract}

\section{Introduction}
One of the central themes of network science is the discovery of peculiar properties that are not exhibited by random or geometric graphs. Over the past decade, network science has built a rich repository of data sets derived
from social network, communication networks, biological data, internet trace data, and more. Early measurements on these networks demonstrated  skewed degree distributions, high clustering coefficients, and community structure~\cite{BarabasiAlbert99,WaSt98,Ne03,Newman-2006-modularity}. These measurements led to fundamentally new mechanisms that explain the networks~\cite{ca-astro,SeKoPi11,Bonato-2014-dimmatch}. Accurately capturing and understanding these properties is critical to understanding the limits of what is possible with rich empirical data in graph-based learning~\cite{SSS20}. 

But many important network quantities are computationally intractable in the worst case and are only computed by heuristics. It is of critical importance to have rigorous theory that can guarantee the accuracy of these measurements. 


\emph{We focus on one of the most significant network characteristics: the cluster structure~\cite{Flake-2000-communities,Newman-2006-modularity,Luxburg-2012-clustering}.} Finding tightly connected sets of vertices with few connections outside is a central task in network analysis. This is often measured by the conductance. The conductance of a set $S$ of vertices is the normalized fraction of edges that leave the set (the normalization is more involved; we give a formal definition later). 

An important development in the cluster structure of real-world networks was the discovery of set size versus conductance relationships~\cite{LLDM08, Leskovec-2009-community-structure, LLM10, Gleich-2012-neighborhoods,Jeub-2015-locally}. 
The key finding in these studies is counter-intuitive: in most real-world datasets, we cannot find \emph{large sets} of \emph{small conductance}. An example of this structure is shown in Figure~\ref{fig:intro}. This finding directly contradicts the behavior of conductance in graphs that are derived from nearest neighbors in a geometry or graphs commonly used in partitioning computational domains, where the smallest conductance values occur in large sets. Moreover, the definition of minimum conductance is typically biased towards large sets (see equation \eqref{eq:conductance}), but real-world networks exhibit the opposite behavior.

The key finding is the behavior of the \emph{network community profile (NCP)}. The NCP plots, for each $s$, the minimum conductance among sets of size  (technically volume) $s$. (Refer to Figure~\ref{fig:intro}.) Observe how the plot (the blue line) slopes upward after an initial dip. This trend is consistent across many real-world networks. The NCP of a typical geometric graph slopes downwards. Currently, the NCPs are generated entirely through principled heuristic computations. Hence, it is difficult to guarantee the characteristic real-world behavior of the NCP curve without appropriate theoretical bounds. Our proposed algorithm is the first that can actually give a lower bound on the minimum conductance at a fixed size $s$.


\begin{tuftefigure}[t]
     \centering
        \includegraphics[width=\linewidth]{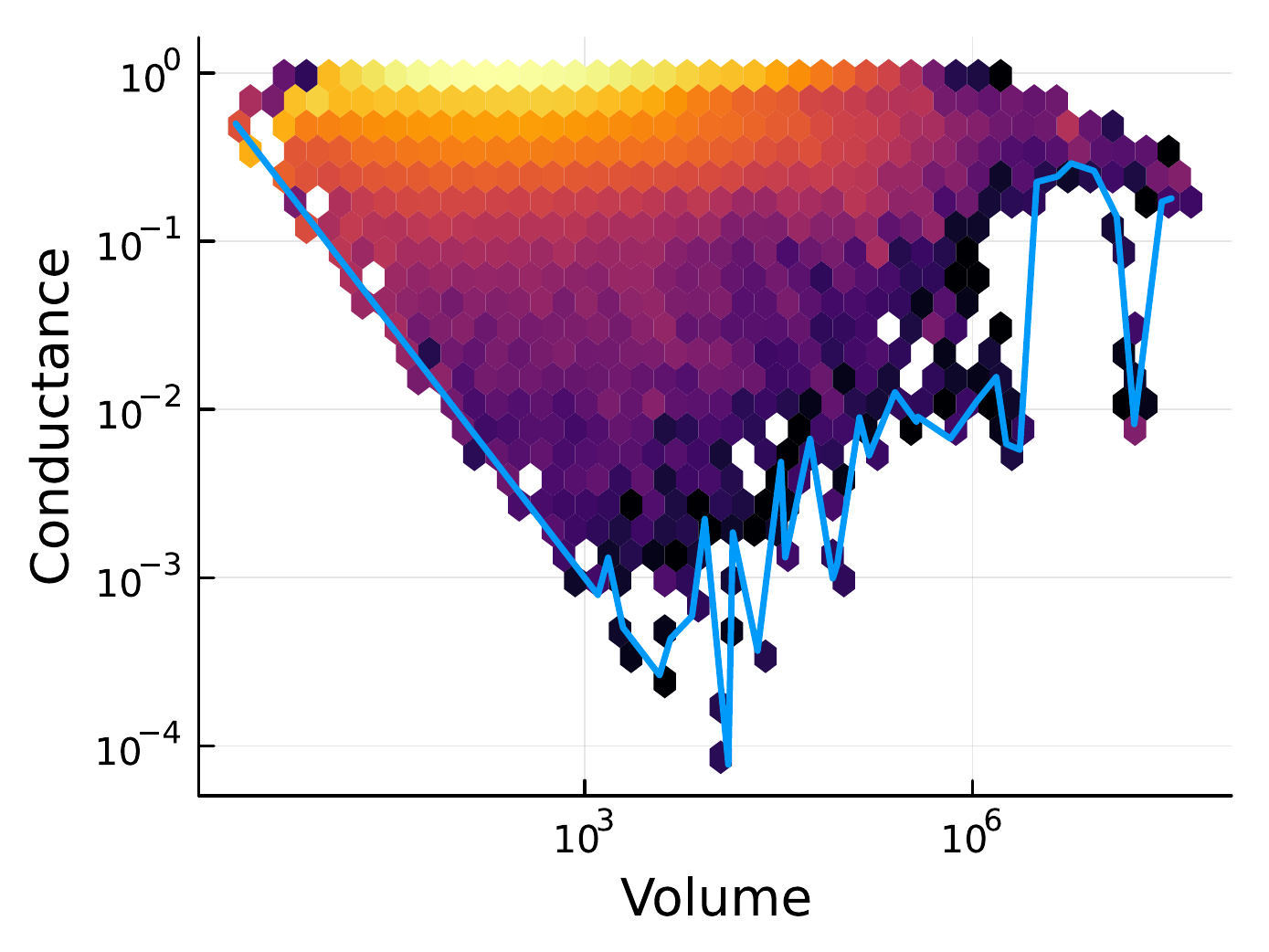}
\caption{
This is a heatmap over size (measured by set volume) and conductance values from  around 630k sets identified by a seeded PageRank conductance minimizing procedure. This picture shows the expected and standard behavior of the size-resolved set conductance in a real-world graph (\textsc{LiveJournal} social network with around 5M vertices and 40M undirected edges). Namely, we see that the smallest conductance sets are those that are also small (volume less than $10^5$ although the graph volume is $10^8$). As sets get larger, their conductance grows. The lower envelope of these measurements is called the network community profile (NCP). The key insight is that for real-world networks, these network community profiles go up and to the right. This is a consistent trend for many real-world social and information networks. 
 A weakness with this empirical finding is that there are no \emph{lower bounds} for seeded PageRank that would certify the behavior of the overall profile. Our paper provides those bounds.}
\label{fig:intro}
\end{tuftefigure}

\begin{center} 
The primary motivating question for our paper is: \emph{can we design theoretically rigorous algorithms that give practically viable bounds for NCP?}
\end{center}

\smallparagraph{Main contributions.} 

\begin{asparaenum}
\item Our main conceptual contribution is a connection between the technical notion of $\mu$-conductance from Markov chain theory~\cite{Lovasz-1990-mixing} and the empirical observations from social and information networks focused of the NCP. We discover that the interesting parts of the NCP basically correspond to a plot of $\mu$-conductance. While this is easy to see (in hindsight), our insight provides us with an array of technical tools to address our main question.
\item We begin with a spectral relaxation to compute the $\mu$-conductance. Unlike the standard relaxation for conductance that leads to the second eigenvalue, the $\mu$-conductance program is non-convex. We give a further convex relaxation using semi-definite programs (SDPs). Unfortunately, this program would require super-quadratic time to solve and is practically infeasible. We give a computationally viable low rank formulation (but non-convex) of the SDP. We prove that locally optimal KKT points of this optimization problem yield rigorous lower bounds for the NCP points. (This is stated in Theorem~\ref{thm:main}, our main theoretical result.) We note that this is first theoretically sound and practically viable lower bound for the NCP.
\item The low rank SDP can be solved with over 200k nodes. Using our algorithms (and Theorem~\ref{thm:main}), we provide the first validation on the shape of the NCP on real-world data, as first discovered in~\citet{Leskovec-2009-community-structure}. Even though our bound is loose, our lower bound validates the characteristic NCP plot and tracks the upward increase in conductance for larger set volume (Figure~\ref{fig:intro-ncps}). 
We are able to distinguish somewhat anomalous graphs such as Deezer (Figure~\ref{fig:deezer}) with a ``flat'' NCP, known to occur in particularly dense real-world networks~\cite{Jeub-2015-locally}. Furthermore, with our new tool, we are able to study, with theoretical confidence, the NCP as the graph is ``peeled'' by $k$-core analysis. 
\end{asparaenum}

\begin{fullwidthfigure}[t]
\centering
     \begin{subfigure}[t]{0.32\linewidth}
         \centering
         \includegraphics[width=\linewidth]{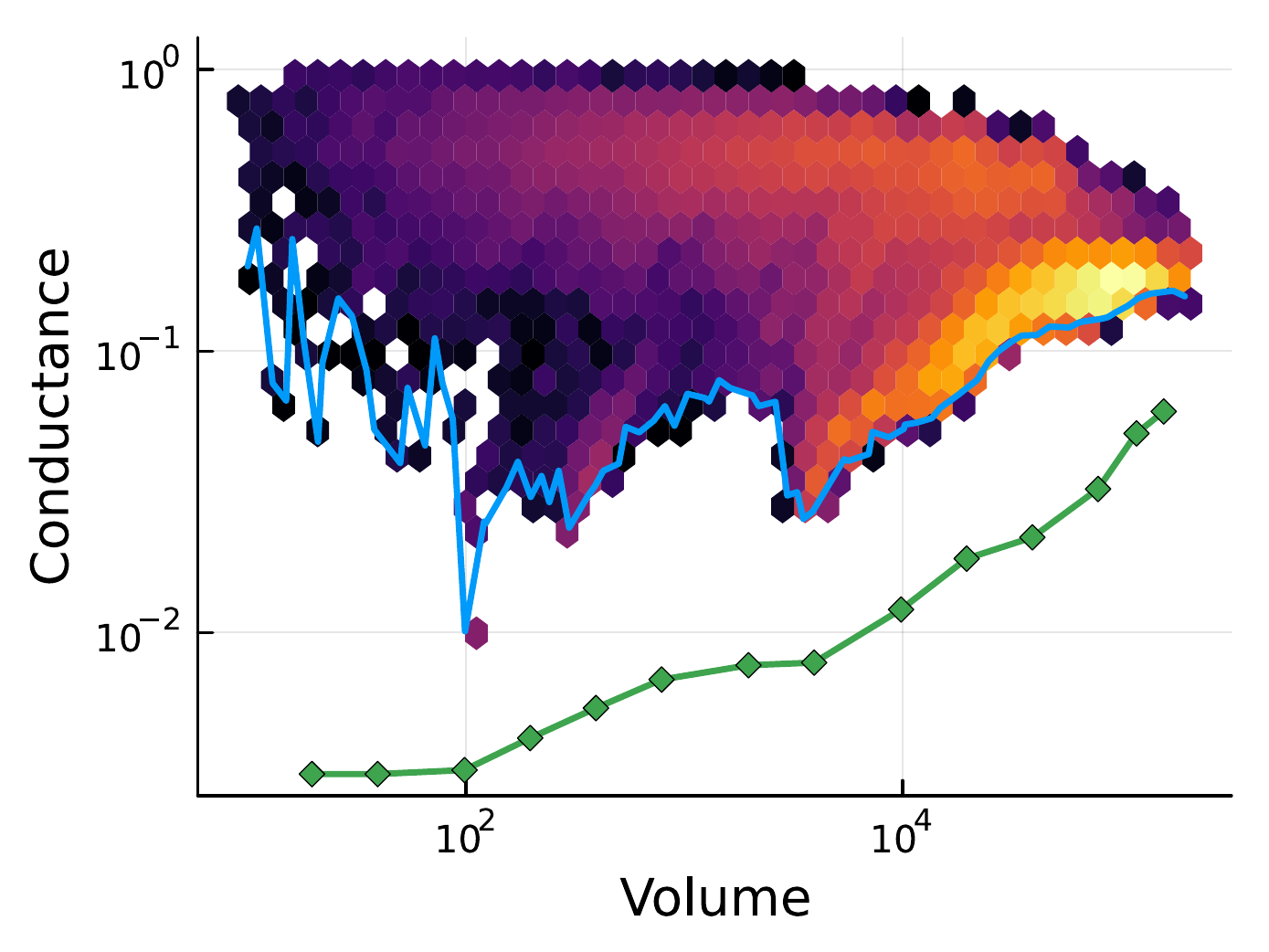}
         \caption{\textsc{AstroPh} }
         \sublabel{fig:astroph}
     \end{subfigure}
     \hfill
     \begin{subfigure}[t]{0.32\linewidth}
         \centering
         \includegraphics[width=\linewidth]{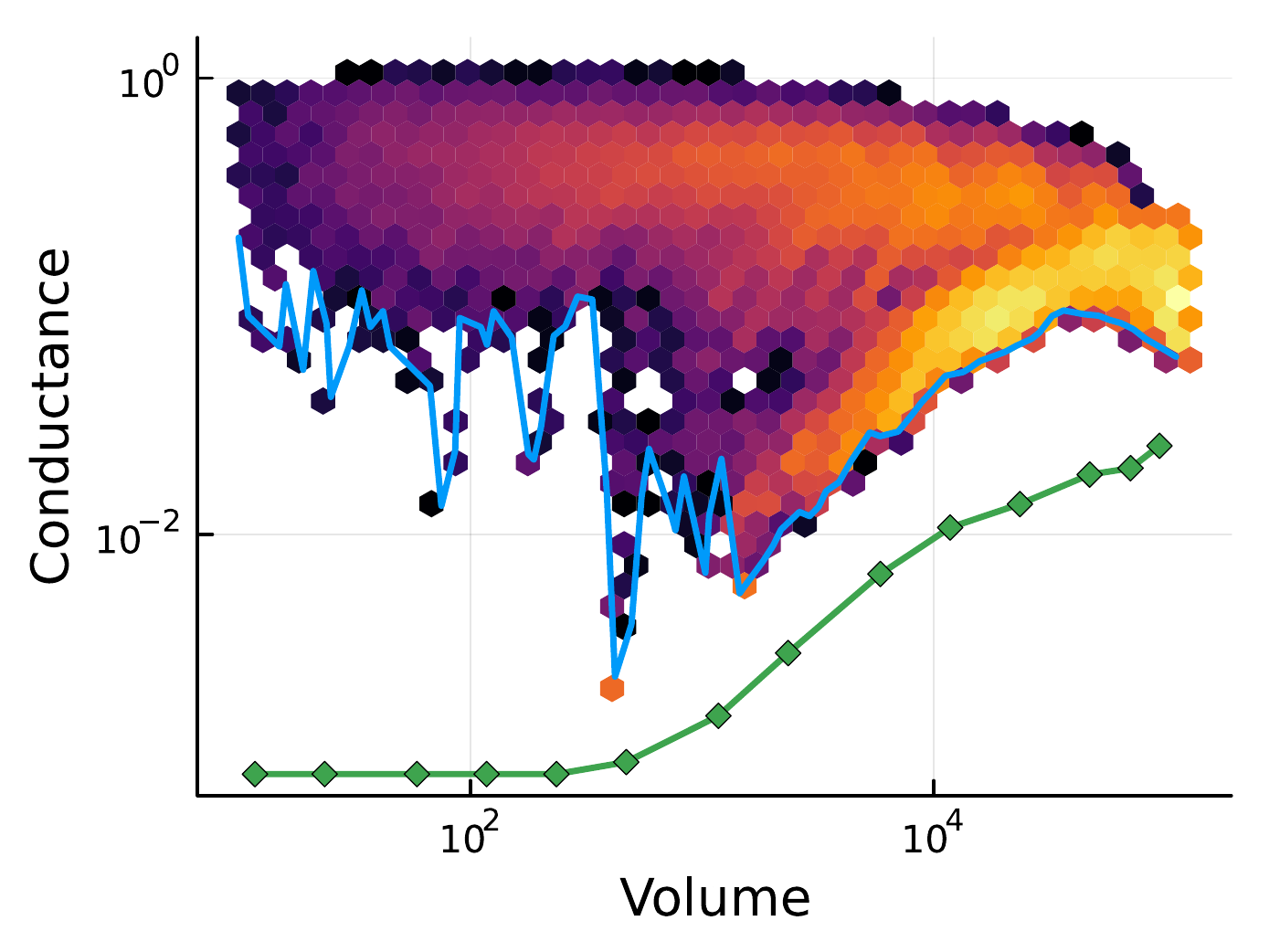}
         \caption{\textsc{HepPh}}
         \sublabel{fig:hepph}
     \end{subfigure}
     \hfill
     \begin{subfigure}[t]{0.32\linewidth}
         \centering
         \includegraphics[width=\linewidth]{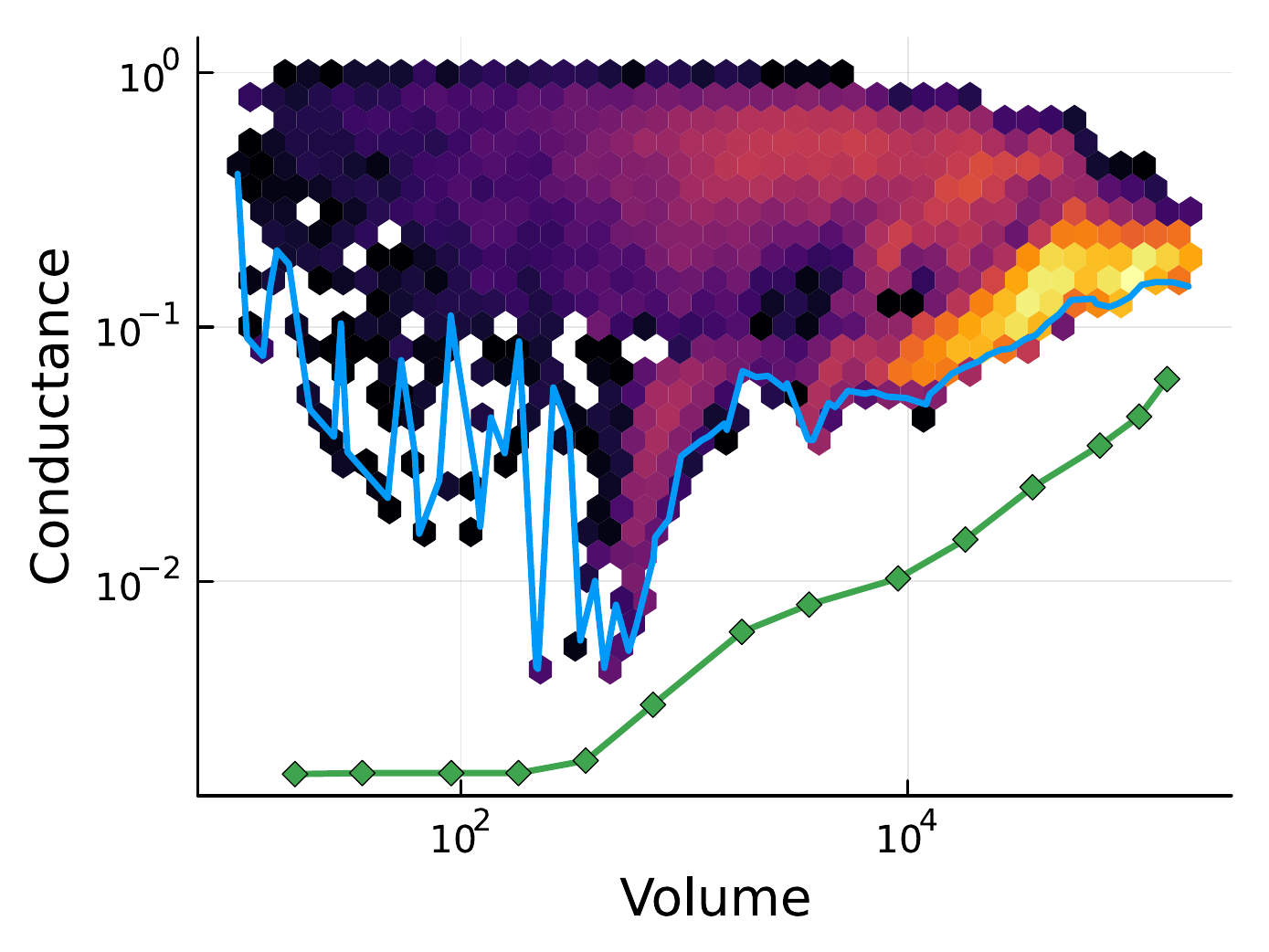}
         \caption{\textsc{email-enron}}
         \sublabel{fig:email-Enron}
     \end{subfigure}
     \\
     \begin{subfigure}[t]{0.32\linewidth}
         \centering
         \includegraphics[width=\linewidth]{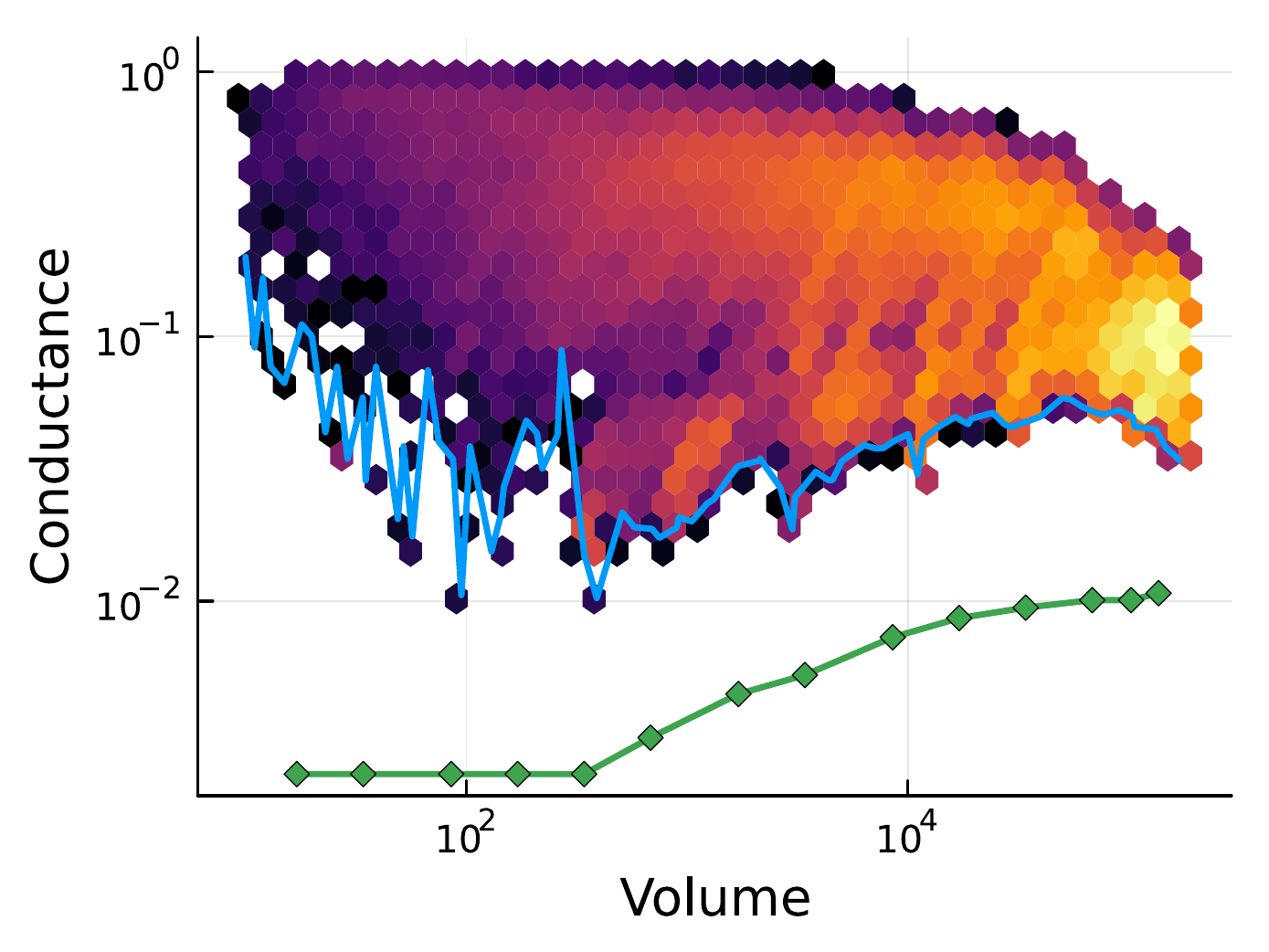}
         \caption{\textsc{facebook-page}}
         \sublabel{fig:facebook}
     \end{subfigure}
     \hfill
     \begin{subfigure}[t]{0.32\linewidth}
         \centering
         \includegraphics[width=\linewidth]{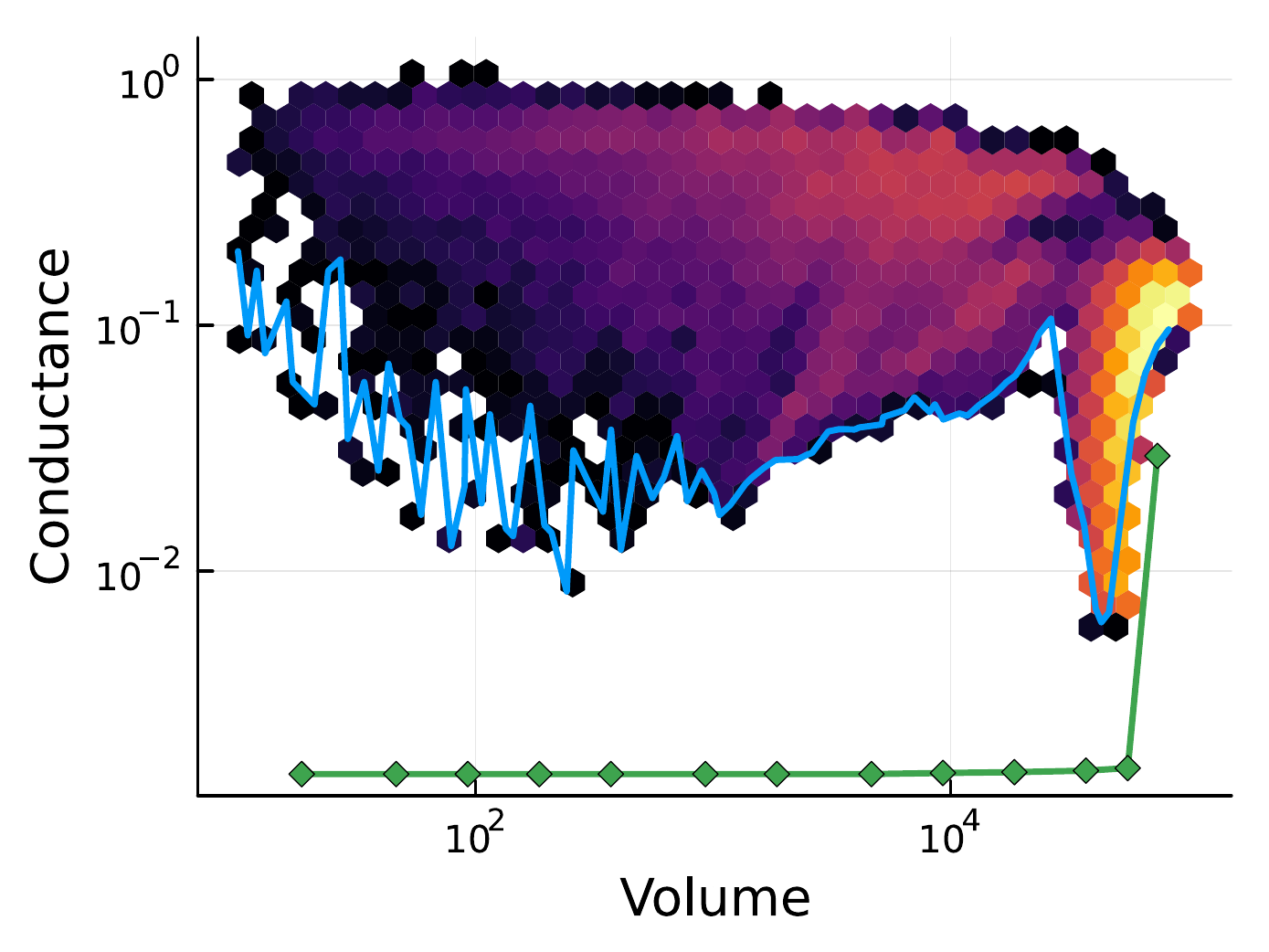}
         \caption{\textsc{deezer}}
         \sublabel{fig:deezer}
     \end{subfigure}
     \hfill
     \begin{subfigure}[t]{0.32\linewidth}
         \centering
         \includegraphics[width=\linewidth]{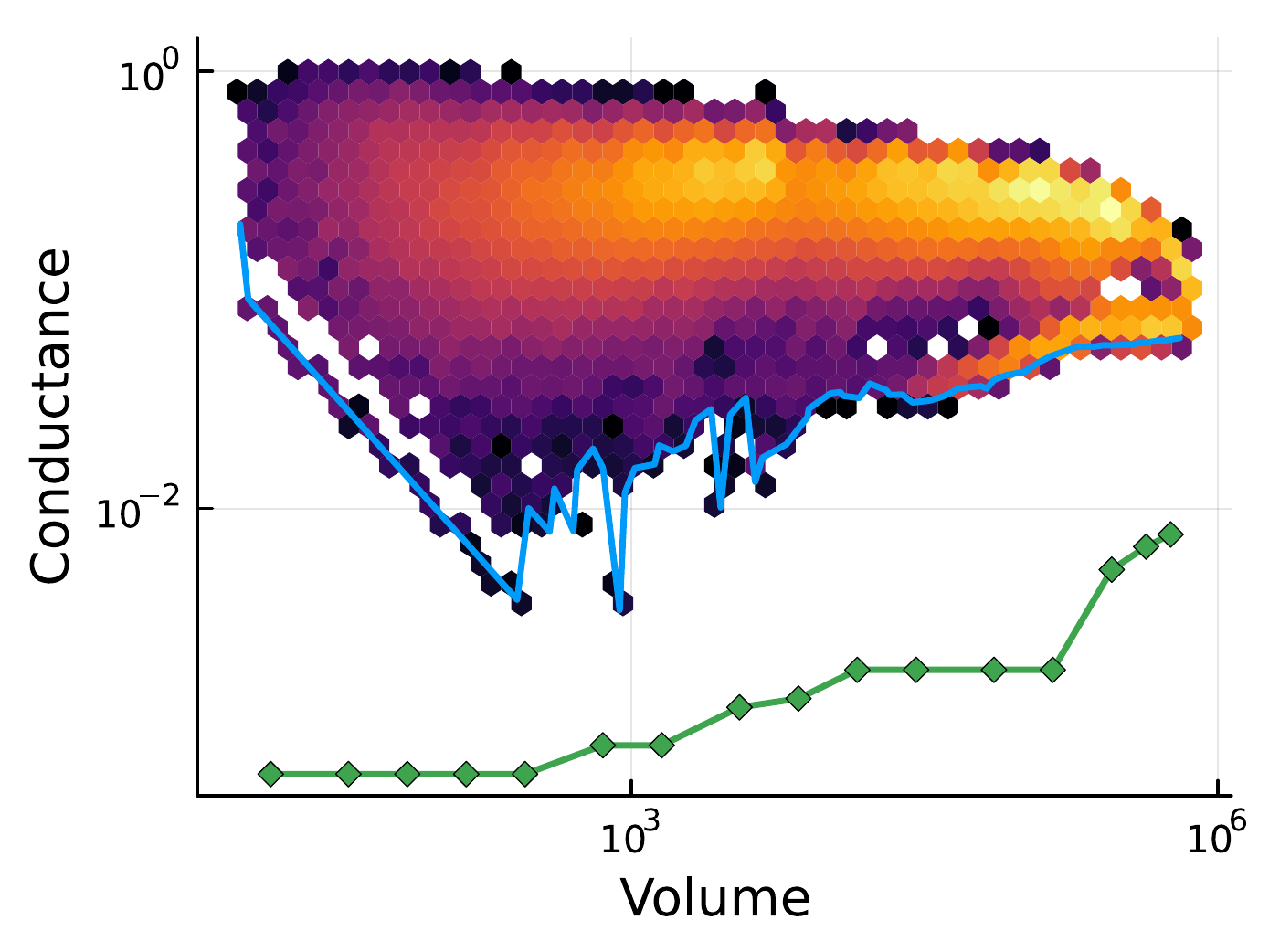}
         \caption{\textsc{DBLP}}
         \sublabel{fig:amazon}
     \end{subfigure}
\caption{Using the theory and algorithms proposed in this paper, we show the empirical network community profile along with our new lower bound for 6 real-world networks (the green line is our lower bound). This is the first guarantee on the behavior of these profiles that establishes a smooth transition from the sets of small conductance to sets of larger conductance. Note that this does not occur for all networks. For example, the Deezer network displays a flat profile until $\mu$ becomes really large, and our results confirm that we should not expect better small conductance sets. The gap between the measured conductances is expected because our analysis only gives a rough, yet informative, lower bound. }
\label{fig:intro-ncps}
\end{fullwidthfigure}

\smallparagraph{High-level outline.} We give a high-level outline of the main ideas in our paper, and explain the chain of theoretical insights that lead to the final practical lower bounds. Our starting point is the notion of 
$\mu$-conductance, discovered in the context of mixing bound for random walks in high dimensional bodies~\cite{Lovasz-1990-mixing}. It is defined formally in equation~\eqref{eq:mu-cond}. Simply put, the $\mu$-conductance value is the minimal conductance restricted to sets with a $\mu$-fraction of the total graph. It was originally proposed to improve the bounds on performance of volume sampling algorithms for convex bodies and study Markov chains where small sets need not have large conductance. 
\emph{One can essentially generate the NCP by computing $\mu$-conductance for varying values of $\mu$ where $\mu$ fixes the volume scale for the sets under consideration.}

\begin{margintable}
\caption{Network Datasets. We report the number of vertices and edges of the largest connected component with self-loops removed.}
\label{tab:dataset}
\begin{sc}
\noindent
\begin{tabular}{@{}l@{\,}cc@{}}
\toprule
Dataset & $|V|$ & $|E|$ \\
\midrule
HepPh & 11,204 & 117,619 \\
AstroPh & 	17,903 &  196,972 \\
Facebook-Page & 22,470 & 170,823 \\
Deezer-EUR & 28,281 & 92,752 \\ 
email-Enron & 33,696 & 180,811 \\
DBLP & 226,413 & 716,460 \\
\bottomrule
\end{tabular}
\end{sc}

\end{margintable}

But, for a given $\mu$, how to study the optimal $\mu$-conductance? A natural starting point is the classic spectral relaxation for the conductance of the graph. The conductance is co-$\mathbb{NP}$ hard to compute, but one can consider a continuous relaxation \eqref{eq:spectral_cut}. This relaxation is convex and the optimal objective is the second eigenvalue, or spectral gap, of the Laplacian. Since the $\mu$-conductance is a constrained version of conductance, we can adapt that constraint into the spectral relaxation. That yields the program \eqref{eq:mu_spectral_cut}. Note that these programs optimize over vectors, rather than sets. The spectral program for $\mu$-conductance has extra constraints bounding each entry of the vector. These extra constraints lead to a non-convex program, showing how computing $\mu$-conductance is significantly harder than conductance.

We now make a further relaxation, wherein we replace the vector by a positive semi-definite matrix. This relaxation leads to the semi-definite program (SDP) given in \eqref{eq:mu_spectral_cut_sdp}. SDPs are convex programs, but the number of variables is quadratic in the number of vertices. This program is infeasible for graphs with even tens of thousands of nodes. 

To get a practically viable optimization problem, we formulate a low rank version of the SDP, stated in \eqref{eq:lowrank_sdp}. But this problem is non-convex and cannot be solved globally. We now arrive at the deepest technical insight in our result. Consider locally optimal KKT points of the low rank (non-convex) SDP. We do a careful comparison of the KKT conditions of the convex program \eqref{eq:mu_spectral_cut_sdp} and the low rank, non-convex \eqref{eq:lowrank_sdp}. We discover that KKT points of \eqref{eq:lowrank_sdp} satisfy all KKT conditions of \eqref{eq:mu_spectral_cut_sdp}, barring one dual feasibility constraint. The violation of this constraint gives a bound on how far the low-rank KKT points are from the original SDP optimum. So, we can subtract out this violation from the objective of the low-rank KKT point, and get a provably correct lower bound on the SDP objective (which is a lower bound on the $\mu$-conductance). 

Our code to solve these problems is available from
\begin{center}
\url{https://github.com/luotuoqingshan/mu-conductance-low-rank-sdp}.
\end{center}


\smallparagraph{Potential implications for random walks.}
Random walks are a central tool in modern network analysis.
A common practice in graph-based learning and embedding is to use a random process to sample a region of the graph~\cite{deepwalk,line, node2vec}. Likewise, there are many results that attempt to estimate quantities based on a random sample of a graph~\cite{LF06, ANK10, ribeiro2010estimating, MB11, RT12, ahmed2014graph}. Many of these results have a theoretical bound that depends on the mixing time of the random walk~\cite{DaKu14, ChDa+16, CH18}, which is bounded by the conductance. As the NCPs show, the minimum conductance may be quite small, but only because of sets of small size. So global properties of the graph might not be affected by such small sets. 
Our new $\mu$-conductance theory suggests that the standard mixing time bounds (based on conductance) may be quite pessimistic when the sampling involves a large set in the graph. It is likely that $\mu$-conductance  gives a better estimate of the mixing time for many applications and studying this is an exciting direction for future work.

\section{Preliminaries and Technical Setting }
\label{sec:background}

Throughout the paper, we work with undirected graphs. Our methods and definitions are applicable to graphs with non-negative edge weights. 
Assume, without loss of generality, the vertices $V$ are labeled from $1$ to $n$. Let $E$ be the set of edges (we assume both $(i,j)$ and $(j,i)$ are in $E$ for an undirected graph). Let $\mA$ be the symmetric adjacency matrix where $A_{i,j}=A_{j,i}$ is equal to the edge weight, or is just $1$ for an unweighted graph, and $A_{i,j}$ is 0 for each non-edge.   Let $S$ be a set of vertices in the graph and let $\Sbar$ be the complement set $\Sbar = V \setminus S$. The notation $\partial S$ indicates is the number of edges (or total edge weight) needed to separate the set $S$ from the rest of the graph: $\partial S = \sum_{(i,j) \in E, i \in S, j \in \Sbar} A_{i,j}$. The notation $\vol(S)$ is the sum of edges involving vertices in $S$: $\vol(S) = \sum_{(i,j) \in E, i \in S} A_{i,j}$. By convention, we set $\vol(G) = \vol(V)$. We write $\ve$ for the vector all ones, so $\vol(G) = \ve^\T \mA \ve$ and $\Tr(\cdot)$ denotes the trace.

The conductance of a set of vertices is
\begin{align}
\label{eq:conductance}
    \phi(S) = \frac{\partial S}{\min \{\vol(S), \vol(\bar{S})\}}.
\end{align}
In principle, minimizing conductance finds sets where $\vol(S)$ is large and $\partial S$ is small. Thus, it is interesting that empirical NCPs suggest that the \emph{best} conductance sets are not the largest. The $\mu$-conductance of a graph is

\begin{equation}
\phi_\mu(G) = \!\begin{array}[t]{l@{\,\,\,}l} \displaystyle \mathop{\text{minimize}}_{S \subset V} & \phi(S) \\
\text{subject to} & \mu \vol(G) \! \le \! \vol(S) \! \le \! \vol(G) / 2  \end{array}
\label{eq:mu-cond}
\end{equation}
Here we adopt a slightly different definition from Lovasz and Simonovits's original paper. The definitions are similar in spirit as they both neglect sets with volume smaller than a specific volume but the original one involves a perturbed conductance. Note that if the set of smallest conductance in the graph $G$ is \emph{large} with $\vol(S) \approx \vol(G)/2$, then there is no difference between the $\mu$-conductance and conductance values. It is only for graphs with hypothetical real-world NCP structure that we expect to see interesting behavior from $\mu$-conductance. 

\subsection{Cheeger inequalities and spectral cuts.}
The Cheeger inequality gives a two-sided bound to the set of best conductance in a graph via an eigenvector computation~\cite{Chung-2007-four,cheeger69}. Our manuscript focuses on lower bounding the conductance of sets, rather than upper-bounding them, so we are only concerned with one side of the Cheeger inequality.  The eigenvector computation uses the Laplacian matrix $\mL = \mD - \mA$, where $\mD$ is a diagonal matrix of row-sums of $\mA$, that is, $\mD = \text{Diag}(\vd)$ where $\vd = \mA \ve$. Formally, let 
\begin{equation}
\tag{Spectral Cut}\label{eq:spectral_cut}
\begin{array}[t]{lll}
\lambda_2 = &\underset{\vx \in \R^V}{\text{minimize}} & \vx^\T \mL \vx \\[1ex]
  & \text{subject to} & \vx^\T \mD \vx = 1\\
  &  & \vx^\T \vd = 0.
 \end{array}
\end{equation}
The value $\lambda_2$ is the second smallest generalized eigenvector of $\mL \vx = \lambda \mD \vx$. 
This eigenvector problem is called a spectral cut because $\vx^T \mL \vx$ computes a cut in the graph that we have relaxed over the space of eigenvectors. It is well-known that 
\[ 
\lambda_2/2 \le \mathop{\text{min}}_{S\subset V} \; \phi(S).
\]

\subsection{Network community profiles.}
Network community profiles are typically computed by running either seeded PageRank~\cite{andersen2006-local}, a flow improvement algorithm~\cite{Lang-2004-mqi,andersen2008-improve}, or a customized procedure~\cite{Gleich-2012-neighborhoods} over a large number of random seeds with parameters designed to explore a variety of set sizes as in~\cite{Leskovec-2009-community-structure,Jeub-2015-locally}. Formally, the NCP is the lower envelop of the size-vs-conductance over all sets in the graph (see the lower bound in Figure~\ref{fig:intro}). We find it useful to display a heatmap over all sets sampled in addition to the lower envelop. 

\subsection{Spectral Profile}
\label{sec:spectral-profile}
\begin{fullwidthfigure}[t]
     \centering
     \includegraphics[width=1.0\linewidth]{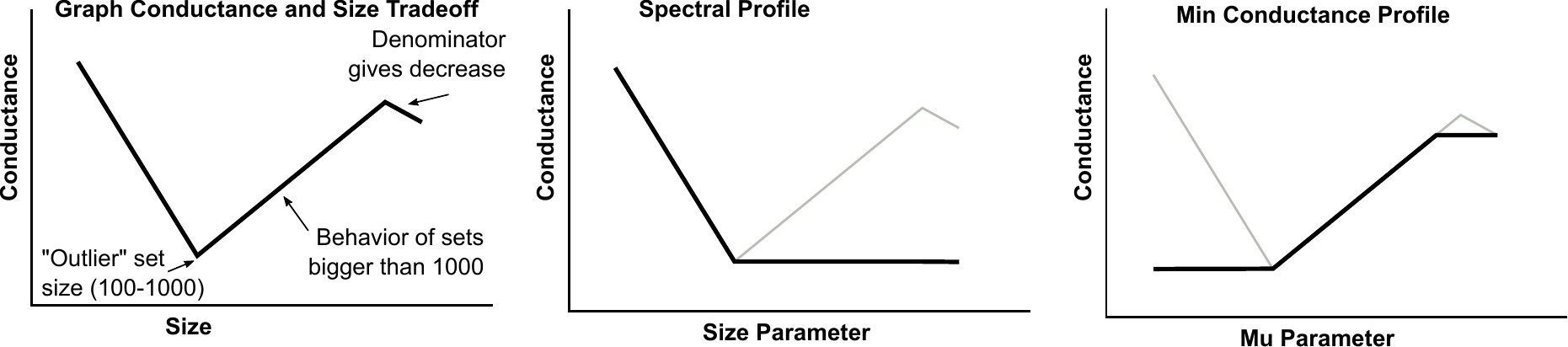}
\caption{Three different notions of a network profile. The left figure shows a sketch of the \citet{Leskovec-2009-community-structure} network community profile, which measures the smallest
conductance of sets with a given volume. The spectral profile \cite{goel2006mixing,Raghavendra:2010:AIS:1806689.1806776} is the dark line in the middle figure that measures relaxations of sets with volume up to $r$ fraction of the maximum volume. The $\mu$-conductance
profile \cite{Lovasz-1990-mixing}, the dark line in the right figure, measures the conductance of sets with at least a $\mu$ fraction of the total volume. }
\label{fig:profile-differences}
\end{fullwidthfigure}
The spectral profile or conductance profile~\cite{goel2006mixing,Raghavendra:2010:AIS:1806689.1806776} is another related graph profile, with a key difference and distinction. These conductance profiles study the behavior of sets with size \emph{up} to a given fraction of the volume. In our notation, 
\begin{equation}
\phi_{\max}^{r}(G) = \begin{array}[t]{ll} \displaystyle \mathop{\text{minimize}}_{S \subset V} & \phi(S) \\
\text{subject to} & \vol(S) \le  r \vol(G). \end{array}
\label{eq:spectral-cond}
\end{equation}
Formally, they study spectral profiles, which are related to eigenvalues, but these are within a factor of $2$ of the conductance values. We illustrate the differences between the profile we are interested in, these spectral profiles, and the network community profile~\cite{Leskovec-2009-community-structure,Jeub-2015-locally} in Figure~\ref{fig:profile-differences}. A network community profile just measures the minimum conductance of sets with a given size (technically we use the volume measure throughout this manuscript), which is \emph{swept} over all possible sizes. (Typically, the given size is taken to be an approximation to make the curve look more smooth.) Here, we have shown a characteristic network community profile as described in~\citet{Leskovec-2009-community-structure} (and illustrated in Figure~\ref{fig:intro}) and annotated it with the features that give rise to the characteristic shape.

\subsection{Balanced Cut}
\label{sec:balanced-cut}
Balanced cut is another common problem that seeks to find a set, or group of sets, that are balanced with respect to the size of the graph. It has traditionally been important in parallel computing where balance implies equally distributed workloads. This is similar to $\mu$-conductance with the size of the vertex set instead of volume. Although this is related to the NCP, the techniques for balanced cut tend to focus on good approximation algorithms. Since many of these techniques give approximation algorithms with unknown or hidden constants, they cannot directly translate into lower bounds. It is likely that a suitable adaptation of our techniques might also give lower bounds for balanced cuts as well. 

\section{Main Theorem}
The main theorem of our paper is a computable and informative lower bound on the $\mu$-conductance of a graph. 
\begin{theorem}
\label{thm:main}
Let $G$ be a connected, undirected graph. Fix $0 \le \mu \le 1/2$. Let $\mY^*$ and $\theta$ be from Algorithm~\ref{alg:main}. Then
 \[ \tfrac{1}{2} (\Tr(\mY^* \mL \mY^*) - \theta \cdot \min \{1, \frac{(1 - \mu) n}{\mu \vol(G)}\}) \le \phi_\mu(G). \]
\end{theorem}
This theorem yields an a posteriori bound as we have no a priori guarantee on the value of $\theta$. In practice, $\theta$ is small, around $10^{-3}$ or $10^{-4}$ in most cases. 

\begin{algorithm}[t]
\begin{algorithmic}[1]
\REQUIRE A graph $G$, a scalar $\mu$, and rank parameter $k$     
\ENSURE A lower bound on $\phi_{\mu}(G)$
\STATE Compute a KKT point of \eqref{eq:lowrank_sdp} (e.g. using an Augmented Lagrangian and LBFGSB as in Section~\ref{sec:methods}).
\STATE Let $\mY^*$ be the solution of \eqref{eq:lowrank_sdp} at the KKT point.
\STATE Let $\theta$ be the value from Lemma~\ref{lem:lowerbound}, found via an eigenvalue computation. 
\STATE \textbf{Return} $\frac{1}{2} (\Tr(\mY^* \mL \mY^*) - \theta \cdot \min \{1, \frac{(1 - \mu) n}{\mu \vol(G)}\})$.
\end{algorithmic}
\caption{MuConductanceLowRankSDPLowerBound}
\label{alg:main}
\end{algorithm}

To prove the main theorem, we work through successive transformations of optimization problems that produce lower bounds on $\mu$-conductance. The first is a spectral program akin to \eqref{eq:spectral_cut}. This is relaxed into a computable SDP. That does not scale to larger problems, and so we translate it into a (non-convex) low-rank SDP. The low-rank SDP can only be locally optimized. Consequently, we derive an a posteriori bound by showing that any local minimizer of the low-rank SDP problem is related to a perturbed SDP. 

\subsection{A spectral program for \texorpdfstring{$\mu$}{mu}-conductance}

The problem \eqref{eq:spectral_cut} is equivalently stated 
$\text{min } \frac{\vx^\T \mL \vx}{\vx^\T \mD \vx} \text{ s.t. } \vx^T \vd = 0$. This form makes a more direct relationship with conductance since if $\vx_S$ is an indicator vector for a set $S$, $\vx_S^\T \mL \vx_S = \partial S$, and $\vx_S^\T \mD \vx_S $ is $\vol(S)$. 
To satisfy $\vx_S^\T \vd = 0$ and $\vx_S^\T \mD \vx_S = 1$, as in~\eqref{eq:spectral_cut} we shift and re-scale $\vx_S$ to  
\begin{align*}
   \vpsi_S =  \sqrt{\tfrac{\vol(G)}{\vol(S) \vol(\bar{S})}}\bigl(\vx_S - \tfrac{\vol(S)}{\vol(G) } \ve\bigr).   
\end{align*}
The problem with spectral cut is that \emph{if} the set of minimal conductance is small, the solution $\vx$ is often highly localized. In order to model $\mu$-conductance, consider a set $S$ with volume about $\mu \vol(G)$ and further consider the scaled and shifted indicator vector $\vpsi_S$ on this set. Then we find that $|x_i| \ge \sqrt{\frac{\mu}{(1-\mu)\vol(G)}}$ and $|x_i| \le \sqrt{\frac{1-\mu}{\mu \vol(G)}}$. 
This suggests that if we expect $\vx$ to indicate a large set, something where $\min(\vol(S), \vol(\bar{S}))$ large, then the
elements of $\vx$ should be small, but not too small, and delocalized.
Thus we add constraints to spectral cut \eqref{eq:spectral_cut} to bound the entries, either the infinity norm or maximum of $\vx$ and to separate small entries around zero. 
This should help spread the mass of $\vx$ over the graph as in Figure~\ref{fig:geometric} (where we look at the solution based on the forthcoming SDP). 
Consequently, we pose the following modified spectral cut
\begin{equation}\label{eq:mu_spectral_cut}
\tag{\(\mu\)-Spectral Cut}
\begin{array}[t]{llll}
\lambda_\mu = &\underset{\vx \in \R^V}{\text{minimize}} & \vx^\T \mL \vx \\[1ex]
  & \text{subject to} & \vx^\T \mD \vx = 1 & \text{\itshape{(a)}} \\ 
  &  & \vx^\T \vd = 0 & \text{\itshape(b)} \\
  &  &\|\vx\|_{\infty}\leq \sqrt{\frac{1 - \mu}{\mu\vol(G)}} & \text{\itshape(c)} \\
  &  &|\vx_i|  \geq  \sqrt{\frac{\mu}{(1 - \mu)\vol(G)}} & \text{\itshape(d)} \\
\end{array}
\end{equation}
In particular, the parameter $\mu$ in this program corresponds to the one in $\mu$-conductance.
Notice that for any set $S$ with volume
less than $\mu \vol(G)$, $\vpsi_S$ is not in the feasible region of \eqref{eq:mu_spectral_cut}. Although the optimal solution of \eqref{eq:mu_spectral_cut} does not have to follow the form of $\vpsi_S$, we believe constraints $(c)$ and $(d)$ will rule out  small and localized sets. 

In addition, we have $\lim_{\mu \to 0^+} \lambda_\mu = \lambda_2$. Even stronger, for all $\mu \le$ some $ \mu^*$, we have $\lambda_\mu = \lambda_2$. Thus as $\mu$ gets close to 0, program \eqref{eq:mu_spectral_cut} simplifies to program \eqref{eq:spectral_cut}. We have the following Lemma which is analogous to ``easy side" of the Cheeger inequality.

\begin{tuftefigure}
     \begin{subfigure}[t]{0.5\linewidth}
         \centering
         \includegraphics[width=\linewidth]{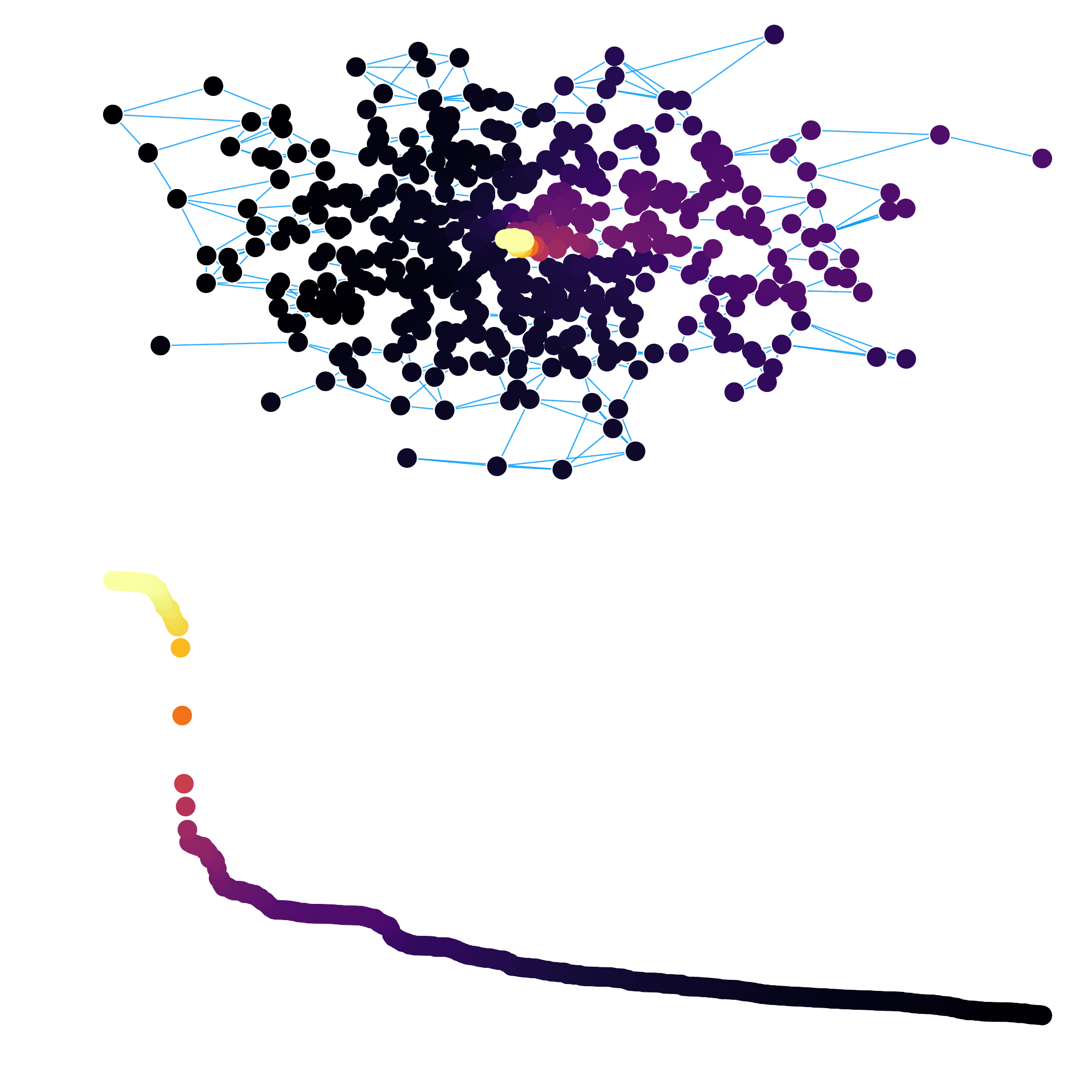}
         \caption{$\mu = 0.03$}
     \end{subfigure}%
     \begin{subfigure}[t]{0.5\linewidth}
         \centering
         \includegraphics[width=\linewidth]{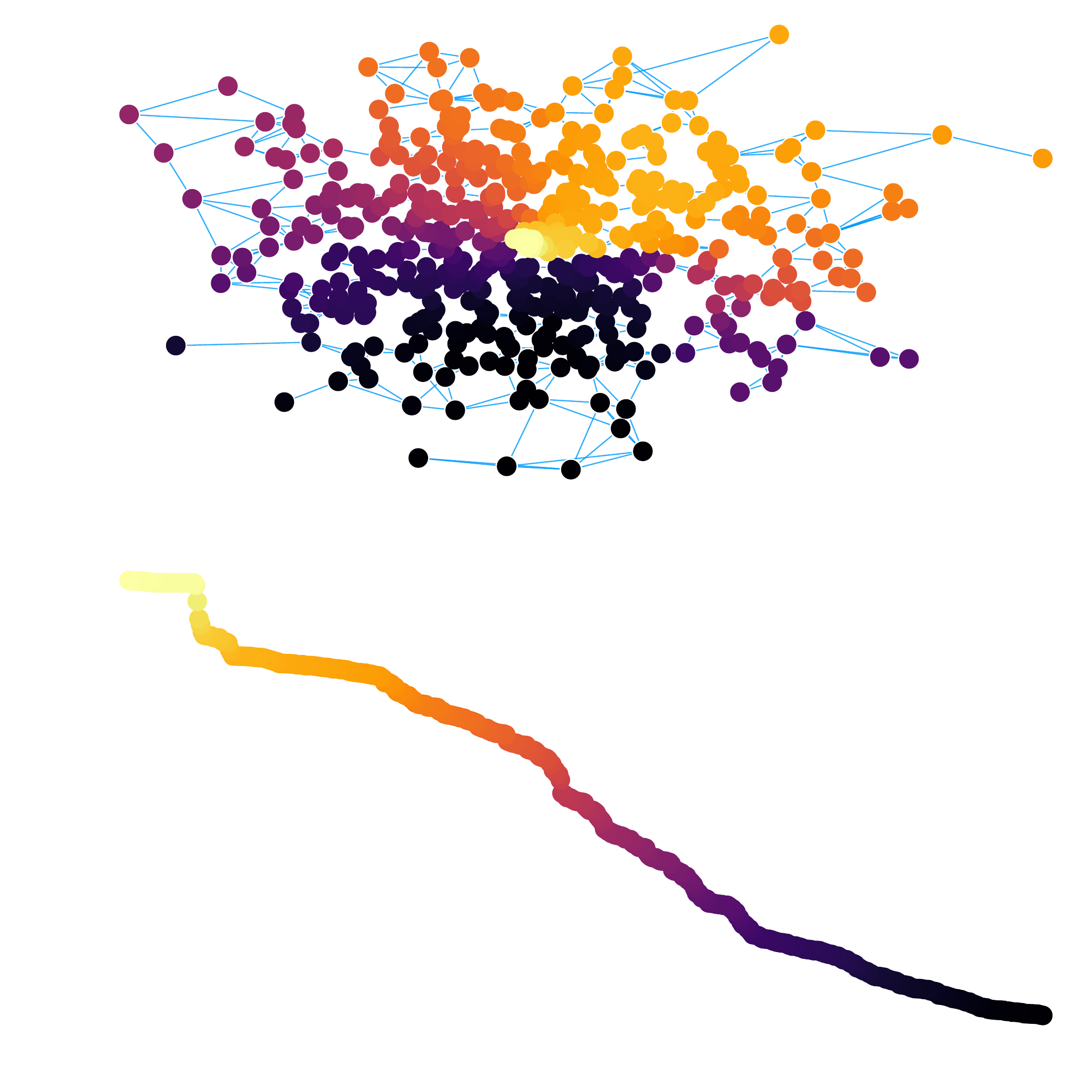}
         \caption{$\mu = 0.3$}
     \end{subfigure}
\caption{We show a vector from the rank-1 approximation of optimal SDP solution $\mX$ on a synthetic graph with 537 nodes and 1327 edges. The vector is shown as colored markers and as sorted values. The graph is constructed to have a small, good conductance set at the center. This shows that as $\mu$ increases, the solution vector delocalizes over the entire network to respond to other sets of reasonably small conductance. }
\label{fig:geometric}
\end{tuftefigure}

\begin{lemma}
\label{lem:spectral-mu-cond-lower}
    For \(0 \le \mu \le 1/2\), 
    $ \frac{\lambda_{\mu}}{2} \le \phi_\mu(G)$.
\end{lemma} 
\begin{proof}
The basic idea is to find a test vector $\vy$ in the feasible region of program 
\eqref{eq:mu_spectral_cut} satisfying $\vy^\T \mL \vy \le 2 \phi_\mu(G)$.
Notice that if $\phi_\mu$ is achieved by the set $T$, then vector 
\[\vpsi_T =  \sqrt{\tfrac{\vol(G)}{\vol(T) \vol(\bar{T})}}\bigl(\ind_T - \tfrac{\vol(T)}{\vol(G) } \ve\bigr)   \] 
is naturally in the feasible region of \eqref{eq:mu_spectral_cut}, where $\ind_T$ is the indicator vector for set $T$. As
\begin{align*}
     \vpsi_{T}^\T \mL \vpsi_{T} 
    &= \frac{\partial T \cdot \vol(G)}{\vol(T) \vol(\bar{T})} \\
    &\leq  \frac{2\partial T}{\min(\vol(T), \vol(\bar{T}))} \\
    &= 2 \phi_{\mu},
\end{align*}
we have $\lambda_{\mu} \leq \vpsi_{T}^\T \mL \vpsi_{T} \leq 2 \phi_{\mu}$.
\end{proof}

Lemma~\ref{lem:spectral-mu-cond-lower} implies that the optimal value of program \eqref{eq:mu_spectral_cut}
can function as a lower bound for the $\mu$-conductance. Furthermore, if we solve program
\eqref{eq:mu_spectral_cut} for different $\mu$s, then the curve of $\lambda^\mu$s 
with respect to corresponding $\mu$ can be a lower bound for the network community profile.

\subsection{A semi-definite program for \texorpdfstring{$\mu$}{mu}-conductance}
We are not aware of any existing techniques to directly solve the problem in the form~\eqref{eq:mu_spectral_cut}. However, it can be relaxed into a semi-definite program (SDP). 
\begin{equation}\label{eq:mu_spectral_cut_sdp}
\tag{\(\mu\)-Cond SDP}
\begin{array}[t]{@{}lll@{}}
\lambda_\mu^{\text{sdp}} = &\underset{\mX \succeq 0}{\text{minimize}} & \Tr(\mL \mX) \\
  & \text{subject to} & \Tr(\mD \mX) = 1 \\
  &  & \Tr(\vd \vd^\T \mX) = 0  \\
  &  &\Diag(\mX) \leq \frac{1 - \mu}{\mu}\frac{\ones}{\vol(G)} \\
  &  &\Diag(\mX) \geq \frac{\mu}{1 - \mu}\frac{\ones}{\vol(G)}.  \\
\end{array}
\end{equation}

The derivation of this relaxation is standard. It follows from replacing $\vx$ with the symmetric positive semi-definite matrix $\mX = \vx \vx^T$ and writing the constraints in an equivalent fashion, then relaxing over all symmetric positive semi-definite matrices. This gives the expected result
\begin{lemma} 
\label{lem:mu-cond-to-sdp}
For \(0 \le \mu \le 1/2 \), we have $\lambda_{\mu}^{\textup{sdp}} \le \lambda_{\mu}$.
\end{lemma}
\begin{proof}
    We verify all steps of the relaxation from \eqref{eq:mu_spectral_cut} to \eqref{eq:mu_spectral_cut_sdp} as the proof of Lemma.

From \eqref{eq:mu_spectral_cut}, let $\vx$ be the variable and let $\mX$  be the rank-1 symmetric positive definite matrix $\mX = \vx\vx^\T$. Then $\vx^\T \mD \vx = 1$ is equivalent to $\Tr(\vx^\T \mD \vx) = \Tr(\mD \vx\vx^\T) = \Tr(\mD \mX) = 1$. Likewise, $\vx^\T \vd = 0$ is equivalent to $(\vx^\T \vd)^2 =\Tr(\vx\vx^\T \vd\vd^\T) = 0$. Finally, for the inequality constraints, $\| \vx \|_\infty \le \alpha$ is equivalent to $x_i^2 \le \alpha^2$ for all $i$, and a similar statement holds for the lower bound on $|x_i|$ too. Thus we arrive at 
\begin{equation*}
\begin{array}[t]{@{}lll@{}}
\lambda_\mu = &\underset{\mX = \vx\vx^\T}{\text{minimize}} & \Tr(\mL \mX) \\
  & \text{subject to} & \Tr(\mD \mX) = 1  \\
  &  &\Tr(\vd \vd^\T \mX) = 0 \\
  &  &\Diag(\mX) \leq \frac{1 - \mu}{\mu}\frac{\ones}{\vol(G)} \\
  &  &\Diag(\mX) \geq \frac{\mu}{1 - \mu}\frac{\ones}{\vol(G)}.  \\
\end{array}
\end{equation*}
Note that this problem is directly equivalent to \eqref{eq:mu_spectral_cut} because of the rank-1 condition $\mX = \vx\vx^\T$. Thus
 we get \eqref{eq:mu_spectral_cut_sdp} and Lemma~\ref{lem:mu-cond-to-sdp} by relaxing the variable $\mX = \vx\vx^\T$ to be a symmetric positive definite matrix. 
\end{proof}

Interestingly, when \(\mu = \frac{1}{2} \), our \eqref{eq:mu_spectral_cut_sdp} is equivalent to the minimium bisection SDP lower bound~\cite{burer2003nonlinear} used in \citet[section 5.2]{Leskovec-2009-community-structure}, which is a previously known lower bound for network community profiles at exactly half the volume. The minimum bisection SDP is
\begin{equation*}
\begin{array}[t]{@{}lll@{}}
\mathcal{C}_G = &\underset{\mY \succeq 0}{\text{minimize}} & \frac14 \Tr(\mL \mY) \\
  & \text{subject to} & \Tr(\vd \vd^\T \mY) = 0  \\
  &  &\Diag(\mY) = \ones. 
\end{array}
\end{equation*}

Formally, we have the following relationship.
\begin{lemma}
\label{lem:equiv-to-min-bisec-sdp}
\[ \frac12 \lambda_{1/2}^{\textup{sdp}} = \frac{2}{\textup{Vol}(G)} \mathcal{C}_G.\]
\end{lemma}
\begin{proof}
    When \(\mu = \frac12\), we notice that the two inequality constraints
    \[\frac{\mu}{1 - \mu}\frac{\ones}{\vol(G)} \leq \Diag(\mX) \leq \frac{1 - \mu}{\mu}\frac{\ones}{\vol(G)}\]
    become the equality constraint \(\Diag(\mX) = \frac{\ones}{\vol(G)}\).
    Further, we can verify that \(\Tr(\mD \mX)\) = 1 is naturally satisfied when 
    \(\Diag(\mX) = \frac{\ones}{\vol(G)}\). Therefore the only difference is scaling. If we let \(\mX =  \vol(G) \mY\) and scale \(\mathcal{C}_G\) by \(\frac{4}{\vol(G)}\), we can see the two programs are exactly the same. Then we get \(\lambda_{1/2}^{\textup{sdp}} = \frac{4}{\vol(G)} \mathcal{C}_G\). 
\end{proof}

\subsection{A low-rank program for \texorpdfstring{$\mu$}{mu}-conductance}

The problem with \eqref{eq:mu_spectral_cut_sdp} is that it has $n^2$ variables,
which means the running time will be worse than $O(n^3)$ in most cases, and may
be $O(n^6)$ in the worst case. Thus it's difficult to get a high-precision solution on graphs with more than a few thousand nodes, which makes it impractical for graphs with tens or hundreds of thousand nodes. However, notice that this program only contains $O(n)$ constraints, thus this program admits an optimal solution with rank at most $O(\sqrt{n})$ \cite{lemon2016low}. This motivates us to change this program to a low-rank SDP formulation via Burer-Monterio method \cite{burer2003nonlinear}. 
Under some mild assumptions, the Burer-Monterio method has good optimality  
and convergence guarantee \cite{BVBNonconvex2016, CifBurer2021, cifuentes2022polynomial}. 
We factorize the positive semi-definite matrix $\mX$ into $\mY \mY^\T$ and introduce slack variables $\vs$  to simplify the inequality constraints to simple bounding box constraints. After these transformations, we arrive at the low-rank program
\begin{equation} \label{eq:lowrank_sdp}
\tag{\(\mu\)-Cond LRSDP}
   \lambda_\mu^{\text{lrsdp}} = \! \begin{array}[t]{@{}l@{\,\,}l@{\,\,}l@{}}
    \displaystyle \mathop{\text{minimize}}_{\mY \in \R^{n \times k}}
        & \Tr(\mY^\T \mL \mY) \\
   \text{subject to} & \Tr(\mY^\T \mD \mY) = 1 & (e) \\
    & \|\mY^\T \vd \|_F^2 = 0 & (f)\\
    &\Diag(\mY \mY^\T) + \vs = \frac{1 - \mu}{\mu\vol(G)} \ones & (g) \\
    &\vs \ge \zeros & (h) \\ 
    &\vs \leq \frac{1 - 2\mu}{\mu(1 - \mu)} \frac{\ones}{\vol(G)}. & (i) 
\end{array}
\end{equation}
Here $k$ is the rank parameter we can tune and we know if $k = \Omega(\sqrt{n})$, then $\lambda_\mu^{\text{lrsdp}} = \lambda_\mu^{\text{sdp}}$.

\subsection{Establishing an overall bound}
\label{sec:overall_bound}

However, the drawback of \eqref{eq:lowrank_sdp} is non-convexity, which makes it hard to be solved globally. Instead we consider 
the KKT points of \eqref{eq:lowrank_sdp}. Since \eqref{eq:lowrank_sdp} is not convex, satisfying KKT conditions of it is no longer sufficient for global optimality. But if we compare the KKT conditions of \eqref{eq:mu_spectral_cut_sdp} and \eqref{eq:lowrank_sdp} closely, we observe that the KKT points of \eqref{eq:lowrank_sdp} directly satisfy all KKT conditions of \eqref{eq:mu_spectral_cut_sdp} except one dual feasibility condition. And the violation of this condition characterizes how far the KKT points of the low-rank program is away from the optimum of the SDP. Formally, let $\lambda \in \R, \beta \in \R, \vgamma \in \R^n, \vg \in \R^n, \vl \in \R^n $ be Lagrangian multipliers corresponding to constraints $(e), (f), (g), (h), (i)$, then we have the following important observation.  
\begin{lemma} \label{lem:lowerbound}
    For a primal-dual pair $\mY^*, \vs^*, \lambda^*, \beta^*, \vgamma^*, \vg^*, \vl^*$ satisfying the KKT conditions of \eqref{eq:lowrank_sdp}, denote 
    \[
    \theta = -\min \{0, \lambda_{\text{min}} (\mL - \lambda^* \mD - \beta^* \vd \vd^\T - \Diag(\vgamma^*)) \}, 
    \] 
    then we have
    \[
        \Tr({\mY^*}^\T \mL \mY^*) - \theta \cdot \min \{1, \frac{(1 - \mu) n}{\mu \textup{Vol}(G)}\} \le \lambda_\mu^{\textup{sdp}}.        
    \]
\end{lemma}
Basically this Lemma states that if the dual variable $\mZ = \mL - \lambda \mD - \beta \vd \vd^\T - \Diag(\vgamma)$ is not positive semi-definite, then we can still lower bound the optimum of the SDP \eqref{eq:mu_spectral_cut_sdp} by subtracting a quantity related to this violation from the objective of \eqref{eq:lowrank_sdp}. 
To prove Lemma~\ref{lem:lowerbound}, we need to first make a few important observations. 

By introducing slack variables $\vs$, we have the following SDP relaxation of \eqref{eq:mu_spectral_cut} which is equivalent to \eqref{eq:mu_spectral_cut_sdp}. 
\begin{equation*}
\begin{array}[t]{@{}lll@{}}
\lambda_\mu^{\text{sdp}} = &\text{minimize} & \Tr(\mL \mX)  \label{eq:musdp} \\
  & \text{subject to} & \Tr(\mD \mX) = 1  \\
  &  &\Tr(\vd \vd^\T \mX) = 0 \\
  &  &\Diag(\mX) + \vs = \frac{1 - \mu}{\mu}\frac{\ones}{\vol(G)} \\
  &  &\zeros \leq \vs \le \frac{1-2\mu}{\mu(1 - \mu)}\frac{\ones}{\vol(G)}  \\
  &  &\mX \succeq 0.
\end{array}
\end{equation*}
Its Lagrangian dual is
\begin{equation}
\tag{$\mu$-Cond SDD} \label{eq:musdd} 
\begin{array}[t]{@{}lll@{}}
\lambda_\mu^{\text{sdd}} = &\underset{\lambda, \beta, \vgamma, \vg, \vl, \mZ }{\text{maximize}} & \lambda + \frac{1 - \mu}{\mu \vol(G)} \vgamma^\T \ve - \frac{1 - 2\mu}{\mu(1 - \mu)\vol(G)} \vl^\T \ve   \\
  & \text{subject to} &\mL - \lambda \mD - \beta \vd \vd^\T -\Diag(\vgamma) - \mZ = \zeros   \\
  &  &\vl - \vg - \vgamma = \zeros \\
  &  &\vl \geq \zeros \\
  &  &\vg \geq \zeros  \\
  &  &\mZ \succeq 0.
\end{array}
\end{equation}
They have the following relation.
\begin{lemma}
    Strong duality holds between \eqref{eq:mu_spectral_cut_sdp} and \eqref{eq:musdd}, in other words, $\lambda_\mu^{\text{sdp}} = \lambda_\mu^{\text{sdd}}$, and the optimum of \eqref{eq:mu_spectral_cut_sdp} is achieved.
\end{lemma}
\begin{proof}
    This is a standard SDP duality claim (for example see \citet{VB96sdp}) implied by the fact that 
    \eqref{eq:musdd} has a strictly feasible solution $\lambda = -1, \beta = -1, \vgamma = \ones, \vl = 2\ones, \vg = \ones$.
\end{proof}
Observe that the objective and all constraints of \eqref{eq:mu_spectral_cut_sdp} are affine with regard to variables $\mX$ and $\vs$, so the KKT conditions are \emph{sufficient} for optimality (see Section~5.5.3 of \citet{boyd2004convex} for example). 

\begin{lemma}
  \label{lem:kktsdp}
  The following KKT conditions are sufficient for a primal-dual pair $\mX^*, \vs^*$ and $\lambda^*$, $\beta^*$, $\vgamma^*$, $\vg^*$, $\vl^*$, $\mZ^*$
  to be an optimal solution.
  The primal feasibility conditions are  
\begin{equation}
\begin{array}[t]{@{}l@{}}
    \Tr(\mD\mX^*) = 1 \\
    \Tr(\vd\vd^\T \mX^*) = 0 \\
    \Diag(\mX^*) + \vs^* = \frac{1 - \mu}{\mu \vol(G)} \ve \\
           \zeros \leq \vs^* \leq \frac{1 - 2\mu}{\mu(1 - \mu)\vol(G)} \ones \\
           \mX^* \succeq 0,
\end{array}
\end{equation}
and dual feasibility conditions are
\begin{equation}
\begin{array}[t]{@{}l@{}}
      \mL - \lambda^* \mD - \beta^* \vd \vd^\T -\Diag(\vgamma^*) - \mZ^* = \zeros\\
             \vl^* - \vg^* - \vgamma^* = \zeros \\
             \vl^* \geq \zeros \\
             \vg^* \geq \zeros \\
             \mZ^* \succeq 0, 
\end{array}
\end{equation}
  and the complementary slackness conditions are  
\begin{equation}
\begin{array}[t]{@{}l@{}}
     {\vg^*}^\T \vs^* = 0  \\
     {\vl^*}^\T (\frac{1 - 2\mu}{\mu(1 - \mu)\vol(G)} \ones  - \vs^*) = 0 \\
     \Tr(\mX^*\mZ^*) = 0.
\end{array}
\end{equation}
\end{lemma}
Note that the stationarity conditions of program \eqref{eq:mu_spectral_cut_sdp} is a subset of the
dual feasibility conditions, so we do not list them out.

Recall that in the low-rank SDP we propose, we just factorize $\mX$ into $\mY \mY^\T$, so it is intuitive it has a strong connection with \eqref{eq:mu_spectral_cut_sdp}. In fact, 
it turns out, for a primal-dual pair $\mY^*, \vs^*$ and $\lambda^*$, $\beta^*$, $\vgamma^*$, $\vg^*$, $\vl^*$ satisfying the KKT conditions of \eqref{eq:lowrank_sdp}, let 
\begin{align*}
\mX^* &= \mY^* {\mY^*}^\T \\
\mZ^* &= \mL - \lambda^* \mD - \beta^* \vd \vd^\T -\Diag(\vgamma^*)
\end{align*}
then $\mX^*, \vs^*$ and $\lambda^*, \beta^*, \vgamma^*, \vg^*, \vl^*, \mZ^*$ are a primal-dual pair which almost satisfies all KKT conditions of \eqref{eq:musdd}, except  
\[\mZ^* \succeq 0.\]

It's easy to verify the claim above because we have the following fact. 
\begin{lemma}
For a primal-dual pair $\mY^*, \vs^*$ and $\lambda^*$, $\beta^*$, $\vgamma^*$, $\vg^*$, $\vl^*$ to satisfy all KKT conditions of \eqref{eq:lowrank_sdp}, they need to satisfy the stationarity conditions
\begin{equation}
\begin{array}[t]{@{}l@{}}
(\mL - \lambda^* \mD - \beta^* \vd \vd^\T - \Diag(\vgamma^*)) \mY^*  = \zeros \\
\vl^* - \vgamma^* - \vg^* = \zeros,
\end{array}
\end{equation}
and primal feasibility conditions
\begin{equation}
\begin{array}[t]{@{}l@{}} \Tr({\mY^*}^\T \mD\mY^*) = 1   \\
  \Tr( \vd \vd^\T \mY^* {\mY^*}^\T) = 0 \\
  \Diag(\mY^* {\mY^*}^\T) + \vs^* = \frac{1 - \mu}{\mu \vol(G)} \ve \\
  \zeros \leq \vs^* \leq \frac{1 - 2\mu}{\mu(1 - \mu)\vol(G)} \ones,
\end{array}
\end{equation}
and dual feasibility conditions
\begin{equation}
\begin{array}[t]{@{}l@{}}
    \vl^* \geq \zeros \\
    \vg^* \geq \zeros,
\end{array}
\end{equation}
  and the complementary slackness conditions  
\begin{equation}
\begin{array}[t]{@{}l@{}}
    {\vg^*}^\T \vs^* = 0  \\
    {\vl^*}^\T (\frac{1 - 2\mu}{\mu(1 - \mu)\vol(G)} \ones  - \vs^*) = 0.
\end{array}
\end{equation}
\end{lemma}

Therefore if $\mZ^* \succeq 0$ is violated, the objective of \eqref{eq:lowrank_sdp} at KKT points may deviate from \(\lambda_{\mu}^{\text{sdp}}\). 

However, we observe that we can bound this deviation by the violation extent of $\mZ^* \succeq 0$. 

Denote
\[
\theta = -\min \{0, \lambda_{\text{min}} (\mL - \lambda^* \mD - \beta^* \vd \vd^\T - \Diag(\vgamma^*)) \}.
\]
If $\theta = 0$, then all KKT conditions of \eqref{eq:mu_spectral_cut_sdp} are satisfied, which means $\mY^*, \vs^*$ achieves global optimality. 

If $\theta > 0$, we consider the following perturbed variant for \eqref{eq:mu_spectral_cut_sdp}.
      \begin{equation}
\begin{array}[t]{@{}lll@{}}
\hat{\lambda}_\mu^{\text{sdp}} = &\text{minimize} & \Tr((\mL + \theta\mI) \mX) \tag{Perturbed $\mu$-Cond SDP} \label{eq:permusdp} \\
  & \text{subject to} & \Tr(\mD \mX) = 1  \\
  &  &\Tr(\vd \vd^\T \mX) = 0 \\
  &  &\Diag(\mX) + \vs = \frac{1 - \mu}{\mu}\frac{\ones}{\vol(G)} \\
  &  &\zeros \leq \vs \le \frac{1-2\mu}{\mu(1 - \mu)}\frac{\ones}{\vol(G)}  \\
  &  &\mX \succeq 0.
\end{array}
\end{equation}
Basically we add $\theta$ into objective and keep feasible region unchanged. 

Denote its dual optimum by \(\hat{\lambda}_{\mu}^{\text{sdd}}\), we can similarly show that strong duality holds, in other words \(\hat{\lambda}_{\mu}^{\text{sdp}} = \hat{\lambda}_{\mu}^{\text{sdd}}\) and \(\hat{\lambda}_{\mu}^{\text{sdp}}\) is achieved.

Now we are ready to prove Lemma~\ref{lem:lowerbound}.
\begin{proof}[Proof of Lemma~\ref{lem:lowerbound}]
For a primal-dual pair $\mY^*, \vs^*$ and $\lambda^*$, $\beta^*$, $\vgamma^*$, $\vg^*$, $\vl^*$ satisfying all KKT conditions of \eqref{eq:lowrank_sdp}, let
\begin{align*}
\mX^* &= \mY^* {\mY^*}^\T \\
\mZ^* &= \mL + \theta \mI - \lambda^* \mD - \beta^* \vd \vd^\T -\Diag(\vgamma^*), 
\end{align*}
then the variables $\mX^*, \vs^*$ and multipliers $\lambda^*, \beta^*, \vgamma^*, \vg^*, \vl^*, \mZ^*$ satisfy all the KKT conditions of \eqref{eq:permusdp} but the following complementary slackness condition is violated
    \[\Tr(\mZ^* \mX^*) = 0, \]
instead we have
    \[\Tr(\mZ^* \mX^*) = \theta \Tr(\mX^*). \]
Since all other conditions are satisfied, we know 
the dual value at this point is 
\[\Tr\bigl((\mL + \theta \mI) \mX^*\bigr) - \Tr(\mZ^* \mX^*) = \Tr(\mL^* \mX^*).\] 
Thus we know 
\[
\Tr(\mL^* \mX^*) \le \hat{\lambda}_{\mu}^{\text{sdd}} = \hat{\lambda}_{\mu}^{\text{sdp}},
\]
which means the objective value at a KKT point of \eqref{eq:lowrank_sdp} is actually 
upper bounded by the optimum of the perturbed SDP \eqref{eq:permusdp}. 

Indeed, we are able to bound the gap between \(\hat{\lambda}_{\mu}^{\text{sdp}}\) 
and \({\lambda}_{\mu}^{\text{sdp}}\).
 Assume \(\mX_{\text{opt}}, \vs_{\text{opt}}\) achieves
the optimum of \eqref{eq:mu_spectral_cut_sdp}. Because the feasible region of \eqref{eq:permusdp} is same with that of \eqref{eq:mu_spectral_cut_sdp},
we know that
\begin{align*}
\hat{\lambda}_{\mu}^{\text{sdp}} \le \Tr((\mL + \theta \mI) \mX_{\text{opt}}) = \lambda_{\mu}^{\text{sdp}} + \theta \cdot \Tr(\mX_{\text{opt}}) \le \lambda_{\mu}^{\text{sdp}} + \theta \cdot \min \{1, \frac{(1 - \mu) n}{\mu \vol(G)}\}, 
\end{align*}
where the last inequality is due to the fact that $\Tr(\mD \mX_{\text{opt}}) = 1$ and $\Diag(\mX_{\text{opt}}) \le \frac{1 - \mu}{\mu \vol(G)} \ones$.

Therefore piecing all things together, we get 
\[
\Tr(\mL^* \mX^*) \le \hat{\lambda}_{\mu}^{\text{sdp}} \le \lambda_{\mu}^{\text{sdp}} + \theta \cdot \min \{1, \frac{(1 - \mu) n}{\mu \vol(G)}\}.
\]
\end{proof}

We remark that the gap $\theta \cdot \min \{1, \frac{(1 - \mu) n}{\mu \vol(G)}\}$
has the potential to be further tightened to get a better posterior bound. The intuition is that assuming $\mX_{\text{opt}} = \mX^*$, then we can turn it into 
$\theta \cdot \Tr(\mX^*)$ where $\mX^*$ is what we know because 
it is $\mY^* {\mY^*}^\T$ and $\mY^*$ is the solution returned by our augmented Lagrangian method. In general, whenever there is some non-trivial relation between trace of $\mX^*$ and $\mX_{\text{opt}}$, we can get a non-trivial tighter bound. We also note that in a further literature review, we found that Lemma~\ref{lem:lowerbound} can be derived from \citet[Theorem 4]{BVBNonconvex2016}.

Summing up all the Lemmas we get, we now have 
\begin{align*}
  \frac12(\Tr({\mY^*}^\T \mL \mY^*) - \theta \cdot \min \{1, \frac{(1 - \mu) n}{\mu \vol(G)}\}) 
 \leq \frac12 \lambda_\mu^{\text{sdp}} 
\leq \frac12 \lambda_\mu
\leq \phi_\mu(G).
\end{align*}
 This concludes the proof of Theorem~\ref{thm:main}. 

\section{Methods}
\label{sec:methods}

In order to solve the non-convex low-rank SDP~\eqref{eq:lowrank_sdp}, we use an augmented Lagrangian approach. The augmented Lagrangian method is an iterative algorithm where in each iteration we minimize a
function including the original objective, the estimated Lagrangian multipliers,
 and the penalty term which drives the solution into feasible region. It has been shown in practice that the augmented Lagrangian method achieves good performance in 
solving low-rank SDP problems \cite{burer2003nonlinear}.

Let $\sigma$ be the coefficient for the penalty term and $\lambda, \beta, \vgamma$ be the Lagrangian multipliers defined in Section~\ref{sec:overall_bound}. The augmented Lagrangian for \eqref{eq:lowrank_sdp} without the bounding box constraint $(h)$ and $(i)$ is
\begin{align*}
 & \mathcal{L}_A(\mY, \vs ; \lambda, \beta, \vgamma,  \sigma) \\
 &\quad = \Tr(\mY^\T \mL \mY) - \lambda(\Tr(\mY^\T \mD \mY) - 1) - \beta (\vd^\T \mY\mY^\T \vd) \\  
 &\qquad - \vgamma^\T (\textstyle \Diag(\mY \mY^\T) + \vs - \frac{(1 - \mu)}{\mu}  \frac{\ones}{\vol(G)} ) \\
 &\qquad + \frac{\sigma}{2} \left( \vphantom{\frac{\mu}{1 - \mu}}  (\Tr(\mY^\T \mD \mY) - 1)^2 + (\vd^\T \mY \mY^\T \vd)^2 \right. \\
 &\qquad \qquad \left. {} + \normof{\textstyle \Diag(\mY \mY^T) + \vs - \frac{(1 - \mu)}{\mu} \frac{\ones}{\vol(G)}}_2^2  \right).
\end{align*}
In each iteration, we solve the following subproblem
\begin{equation} 
\label{eq:alm-subprob}
    \begin{array}{ll}
  {\displaystyle \mathop{\text{minimize}}_{\mY, \vs}} &\mathcal{L}_A(\mY, \vs; \lambda, \beta, \vgamma, \sigma) \\[1ex]
   \text{subject to} & \zeros \leq \vs \leq \frac{1 - 2\mu}{\mu(1 - \mu) } \frac{\ones}{\vol(G)}\\
\end{array}
\end{equation}
using a Limited-Memory BFGS method with bound constraints on variables \cite{byrd1995limited}. Since L-BFGS-B is a quasi-Newton Method, it requires 
the gradient of $\mathcal{L}_A$ with regard to variables $\mY$ and $\vs$. Let 
\begin{align*}
  \vu = \Diag(\mY \mY^\T) + \vs - \frac{\mu}{(1 - \mu) \vol(G)} \ones,   
\end{align*}
we have
\begin{align*}
    \nabla_{\mY} \mathcal{L}_A &= 2 \mL \mY - 2 (\lambda - \sigma (\Tr(\mY^\T \mD \mY) - 1)) \mD \mY \\ 
    &- 2 ( \beta - \sigma \vd^\T \mY \mY^\T \vd) \vd \vd^\T \mY  \\
    &- 2 \left((\vgamma - \sigma \vu ) \ve^\T \right) \circ \mY, \\
    \nabla_{\vs} \mathcal{L}_A &= - \vgamma + \sigma \vu
\end{align*}
where $\circ$ is the element-wise or Hadamard product.

After each solve, we update the multipliers and penalty parameters following Alg~17.4 of~\citet{NW99numopt}. 

\smallparagraph{Initialization and the rank parameter $k$.}
As L-BFGS-B is a quasi-Newton method,
convergence is faster when the starting point is close to the optimal solution. 
We initialize $\mY$ by the $k$ eigenvectors corresponding to the $k$ smallest non-zero eigenvalues
of normalized Laplacian $\mD^{-1/2} \mL \mD^{-1/2}$. This is based on the observation that
when $k = 1$, program \eqref{eq:lowrank_sdp} degenerates to program \eqref{eq:mu_spectral_cut} 
and the Fiedler vector remains the optimal solution for small $\mu$.


\smallparagraph{Comparison against SDP solvers.}
For small enough problems, we can solve both the SDP \eqref{eq:mu_spectral_cut_sdp} as well as the low-rank SDP~\eqref{eq:lowrank_sdp}. 
Therefore we compare our LRSDP with them on 
two small synthetic graphs with 85 and 537 vertices. 
We intentionally construct the  two synthetic graphs with a dense core that has minimal conductance and localizes the Fiedler vector. (See Figure~\ref{fig:geometric} and discussion of the construction in Appendix~\ref{app:synthetic-graphs}.) 
We compare the solvers for different $\mu$s on each graph and the results are summarized in Table~\ref{tab:lrsdp_vs_sdp}. These show that our LRSDP has objective values extremely close to solving the SDPs directly and is \emph{much} faster.

\begin{fullwidthtable}[t]
\centering
\begin{sc}
\begin{small}
\caption{To validate that the LRSDP~\eqref{eq:lowrank_sdp} and SDP~\eqref{eq:mu_spectral_cut_sdp} are similar on problems where we can compute both, we examine their objective values on two small synthetic graphs (e.g.~Figure~\ref{fig:geometric}). We choose two established SDP solvers, SCS \cite{SCS} and Mosek \cite{mosek}.
 This shows that LRSDP gives nearly identical results and is much faster. Here LB stands for the lower bound provided by our low-rank program, which theoretically should be a lower bound for objective value of all SDP solutions. Empirically some objective values are lower than this bound because numerically they do not strictly satisfy all primal feasibility conditions. } 
\begin{tabularx}{\linewidth}{lllXXXXXXX}
\toprule 
Nodes & Edges & $\mu$ & \smash{\rlap{Objective~Value}} & & & Bound &  Time & & \\
 \cmidrule(l){4-6}
 \cmidrule(l){7-7}
 \cmidrule(l){8-10}
 & & & LRSDP & SCS & MOSEK & LB & LRSDP & SCS & MOSEK  \\
 \midrule
\multirow{3}{*}{85} & \multirow{3}{*}{193} 
& 0.01 & 0.004407 & 0.004407 & 0.004406 & 0.004398 & 0.7s & 16.7s & 5.3s \\
& & 0.05 & 0.004510 & 0.004511 & 0.004508 & 0.004499 & 2.1s & 18.4s & 4.9s \\
& & 0.25 & 0.007318 & 0.007223 & 0.007314 & 0.007292 &  1.8s & 18.2s & 6.0s \\
\midrule 
\multirow{4}{*}{537} & \multirow{4}{*}{1327} 
&  0.01 & 0.001092 & 0.001089  & 0.001081 & 0.001083 & 17.8s & 1.6 hrs  & 16.9 hrs \\
& & 0.03 & 0.001115 & 0.001113 & 0.001092 & 0.001056 & 17.3s & 12.2 hrs & 15.6 hrs \\
& & 0.1 & 0.001444 & 0.001440 & 0.001428 & 0.001390 & 21.0s  & 56.7 mins  & 13.4 hrs \\
& & 0.3 & 0.002733 & 0.002732 & 0.002731 & 0.002720 & 11.8s  & 1.8 hrs  & 18.8 hrs \\
 \bottomrule 
\end{tabularx}

\label{tab:lrsdp_vs_sdp}
\end{small}
\end{sc}
\end{fullwidthtable}

\smallparagraph{Non-monotonic results.}
The results from $\mu$-conductance must be monotonic. That is, for $\mu_1 \ge \mu_2$, we must have $\phi_{\mu_1} \ge \phi_{\mu_2}$ by set inclusion properties of the $\mu$-conductance function. Because we have a lower bound, we found scenarios where the lower bounds were not monotonically increase in $\mu$. Since our investigations typically involve multiple values of $\mu$, we simply adjust the bounds to reflect the \emph{tightest} lower bound from any value of $\mu$ that we computed. Practically, this corresponds to taking a stepwise maximum over the experimental results. 

\section{Experiments}

In this section, we revisit the lower bounds from the introduction (Figure~\ref{fig:intro-ncps}). 
We then explore how the \emph{running time} of our programs is affected by graph size, $\mu$, and rank parameter $k$. 
Further, although directly tracking the true NCP is co-$\mathbb{NP}$ hard, we are still able to study the gap between the true NCP and our lower bound by a squeeze bound or \emph{gap shrinking} analysis. 
In the end we do one interesting
$k$-core analysis on one graph using our algorithm, which reveals
the potential for use in other network analysis tasks.

\enlargethispage{\baselineskip}
\subsection{Computational details.} 
When solving the subproblem \eqref{eq:alm-subprob} in our augmented Lagrangian procedure, we use L-BFGS-B with $m=3$. We set the default tolerance of stationarity condition and primal feasibility condition of our augmented Lagrangian as $10^{-5}$. For each dataset, we pick a set of $\mu$s varying from $10^{-6}$ to $0.4$ which is dense enough to form an informative lower bound curve.
 We exhaustively try $k$ from $\{1, 3, 5, 10\}$.
To generate the NCP plots, we empirically sample a large number of sets from a seeded PageRank based method~\cite{andersen2006-local}.
 Specifically, we randomly sample a large collection of seeds and then try different $\varepsilon$ ranging from $10^{-2}$ to $10^{-8}$.
 For each seeded PageRank we get, we perform a sweepcut to get several
 sets with good $\mu$-conductance.

\subsection{Summary of key findings from introduction.}
The main figure for our experiments is Figure~\ref{fig:intro-ncps}. This shows the lower bounds on the NCPs produced by our procedure. We test our procedure on AstroPh , HepPh \cite{ca-astro}, Email-Enron \cite{Leskovec-2009-community-structure}, Facebook-Page \cite{facebook}, Deezer \cite{deezer} and DBLP \cite{BoVWFI, BRSLLP}. Their sizes are in Table~\ref{tab:dataset}. We can see although there is a gap between
our lower bound and the NCP generated by seeded PageRank,  our algorithm provides an informative lower bound which mirrors the trend of the NCP plot.

\subsection{Running time}
\label{app:exp-running-time}
We illustrate the effects of graph size, $\mu$, and rank parameter $k$ on the running time of our program. The results are summarized in Table~\ref{tab:running-time}. We can clearly observe that with 
graph size or rank parameter increasing, the running time increases but roughly linearly. This is expected because each iteration of L-BFGS-B takes linear time with regard to number of variables. Also, we can see with $\mu$
increasing, the running time tends to increase as well.  The intuition 
is that with $\mu$ increasing, the feasible region of the low-rank program shrinks, which makes optimization harder. Also, our initialization favors smaller $\mu$.

\begin{tuftetable}[t]
\centering
\begin{sc}
\begin{small}
\caption{This table summarizes the running time on two graphs with 
 a few different $\mu$ and $k$ choices. We report the running time of
 the augmented Lagrangian method (ALM) for solving low-rank SDP and eigenvalue computation (EIGVAL) for calculating the dual feasibility violation separately.}
\begin{tabularx}{0.8\linewidth}{lXX>{\raggedleft}XX}
\toprule 
Graph & $\mu$ & $k$ & Time &  \\
\cmidrule{4-5}
& &  & ALM & EIGVAL  \\
 \midrule
\multirow{6}{*}{\shortstack{\textsc{HepPh} \\ $|V| = 11204$ \\$|E| = 117619$}}
& \multirow{3}{*}{0.001} & 3 & 1.7 min & 30.8 s \\
& & 5 & 3.4 min & 48.1 s \\
& & 10 & 6.2 min & 36.6 s \\
\cmidrule{2-5}
& \multirow{3}{*}{0.1} & 3 & 1.8 hrs & 21.4 min \\
& & 5 & 3.1 hrs & 12.9 min\\
& & 10 & 6.9 hrs & 23.0 min\\
\midrule
\multirow{6}{*}{\shortstack{\textsc{DBLP} \\ $|V| = 226413$ \\ $|E| = 716460$}}
& \multirow{3}{*}{0.001} & 3 & 21.8 hrs & 3.1 hrs \\
& & 5 & 1.8 days & 3.5 hrs \\
& & 10 & 2.6 days & 8.7 hrs \\
\cmidrule{2-5}
& \multirow{3}{*}{0.1} & 3 & 1.6 days & 1.9 hrs \\
& & 5 & 3.4 days & 33.6 mins \\
& & 10 & 3.1 days & 5.6 hrs \\
 \bottomrule 
\end{tabularx}

\label{tab:running-time}
\end{small}
\end{sc}
\end{tuftetable}

\subsection{Synthetic graph construction}
\label{app:synthetic-graphs}
The synthetic graphs we use are designed to have a dense core with a geometric like periphery. To do this, we first randomly pick coordinates of $n$ points according to 
a normal distribution in each dimension. Then we scale the coordinates of $90\%$ of them by 1.5 and scale the coordinates of those that remain ($10\%$)  by 0.1. This forms a small dense piece at the center. In the end we link each point to its five geometrically closest neighbors. 

\subsection{Gap shrinking.}
Our theory gives a lower bound on the $\mu$-conductance scores. 
To study how close our lower bound can be to the true $\mu$-conductance score which is co-$\mathbb{NP}$ hard to compute, we study how close an upper bound of the $\mu$-conductance score can be to our lower bound. This kind of squeeze bound gives an indirect way to estimate the real gap.  In order to explore how tight our lower bound is, in other words how small the gap can be, for the HepPh graph, we dramatically increase the number of samples of sets from seeded PageRank. The results are summarized in Figure~\ref{fig:gap-shrinking}. We see that our lower bound is not that loose: about $1/3$ off.

\begin{marginfigure}[-40ex] 
     \centering
        \includegraphics[width=\linewidth]{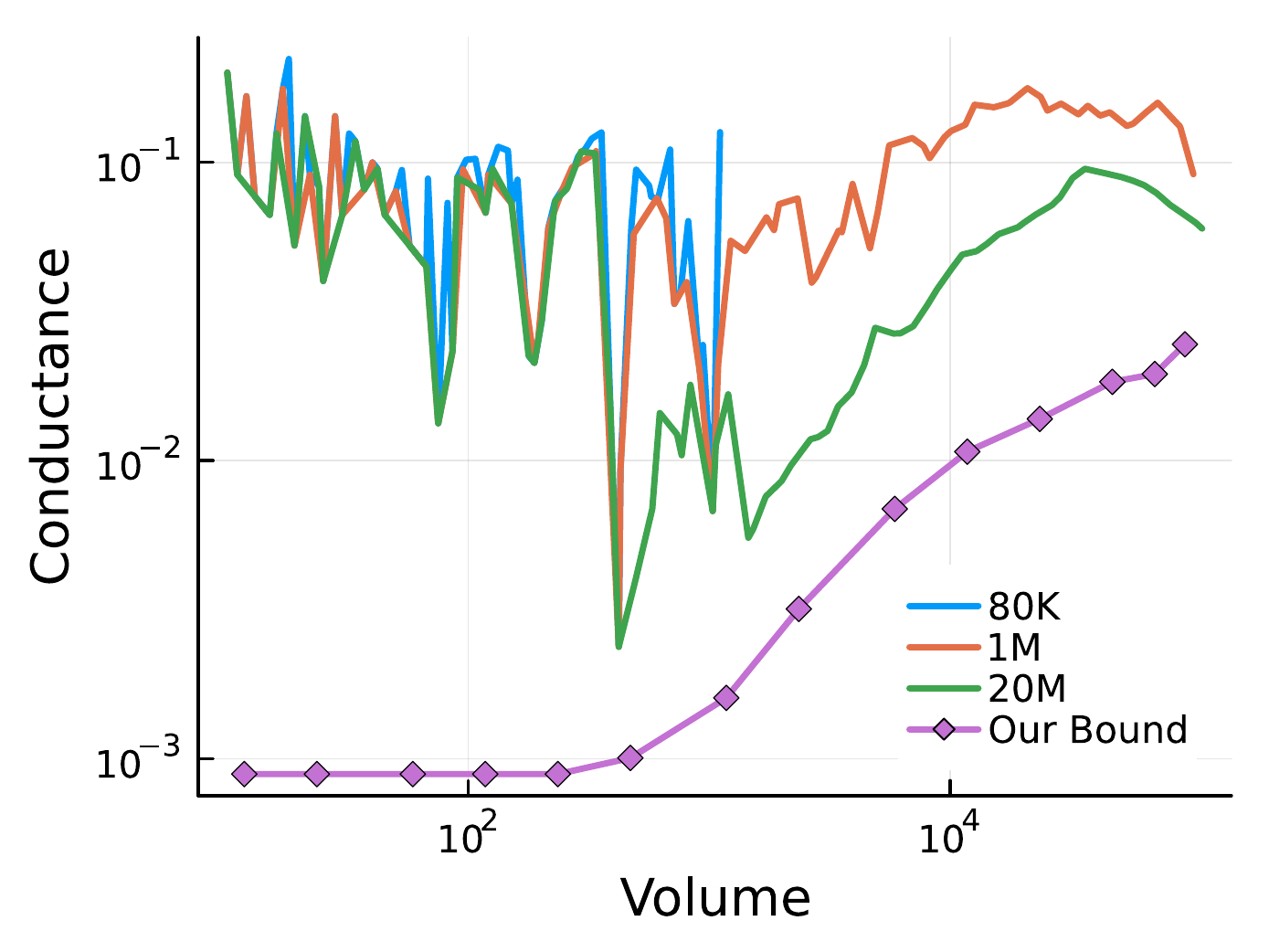}
\caption{Gap shrinking effect illustrated on HepPh. The upper three line plots are the NCPs determined by different number of sets. This shows the more sets 
we search using seeded PageRank, the smaller the gap between the NCP and our lower bound.  }
\label{fig:gap-shrinking}
\end{marginfigure}

\subsection{Investigation with k-cores.}
\begin{marginfigure}
     \centering
        \includegraphics[width=\linewidth]{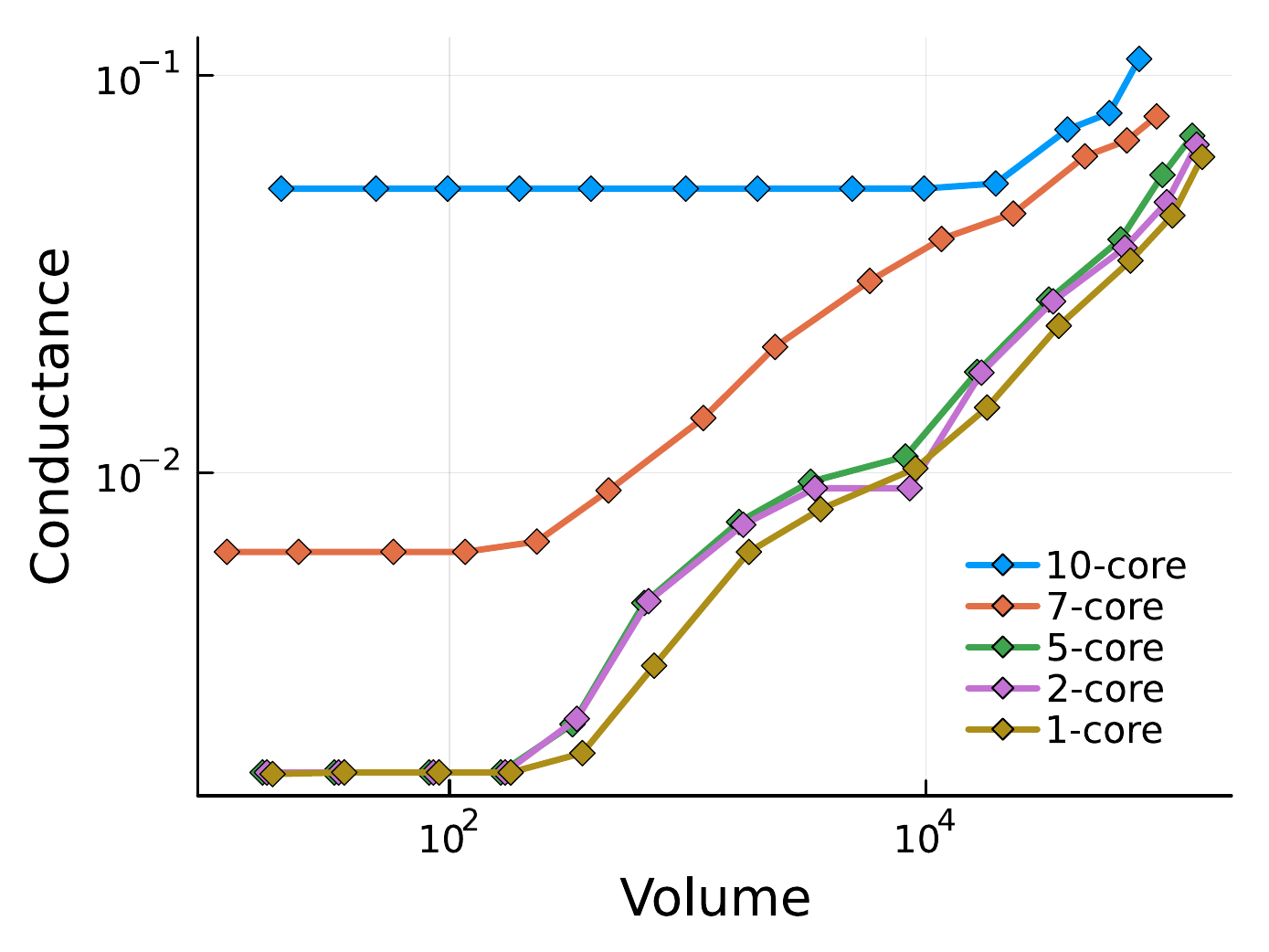}
\caption{$k$-core analysis on Email-Enron.}  
\label{fig:k-core}
\end{marginfigure}
In order to show the potential of applying our method to broader network analysis tasks, we apply our low-rank program to analyze the NCP of $k$-cores of a graph~\cite{Seidman1983-cores}. The inspiration for this study is a discussion over whether the NCP represents a signal or noise mode of a graph~\cite{NEURIPS2018_2a845d4d}. The core number of a vertex in a graph is the largest integer $k$ such that the process of repeatedly removing vertices with degree less than $k$ will not delete this vertex from the graph. So the $1$-core is the entire graph. The $2$-core is there result of sequentially deleting all degree 1 nodes. By analyzing the NCP of $k$-cores with various $k$, we can have a deeper understanding of the structure of a network. The results are summarized in Figure~\ref{fig:k-core}. These show that the NCP structure is preserved for Email-Enron up through the 5-core and is largely preserved at the 7-core. While this single experiment does not to resolve the question of signal vs.~noise for the NCP, it does show how our tools could be used to study it.

\subsection{Comparison with other lower bounds.}
Besides our \(\mu\)-conductance lower bound, there are two previously known lower bounds for
network community profiles mentioned in \cite{Leskovec-2009-community-structure}, 
one spectral bound induced by Cheeger inequality and the Fiedler vector that is independent of volume and 
the other is given by the minimum bisection SDP. 
To get a comprehensive understanding of how our lower bound behaves compared with
existing lower bounds we compare on two graphs. As is shown
in Lemma~\ref{lem:equiv-to-min-bisec-sdp}, the minimum bisection SDP lower bound is actually equivalent to ours at 
\(\mu = \frac12\), here we directly solve our low-rank SDP at \(\mu=\frac12\) instead of 
solving the minimum bisection SDP. The results on AstroPh and HepPh graphs are shown in 
Figure~\ref{fig:lowerbound_comp}. These show that we smoothly interpolate between the bounds as expected.

\begin{tuftefigure}[h]
    \centering
        \includegraphics[width=0.45\linewidth]{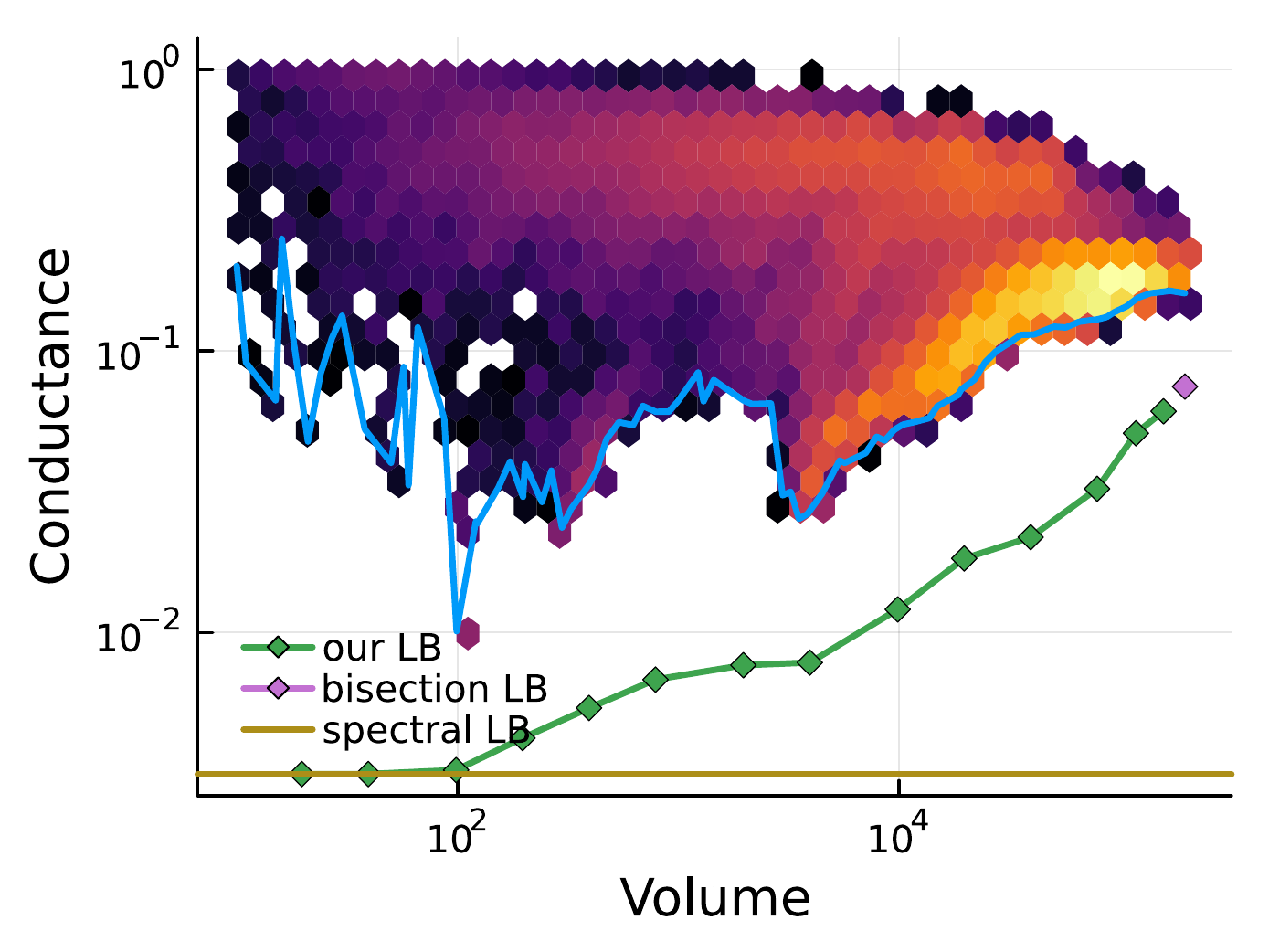}
        \includegraphics[width=0.45\linewidth]{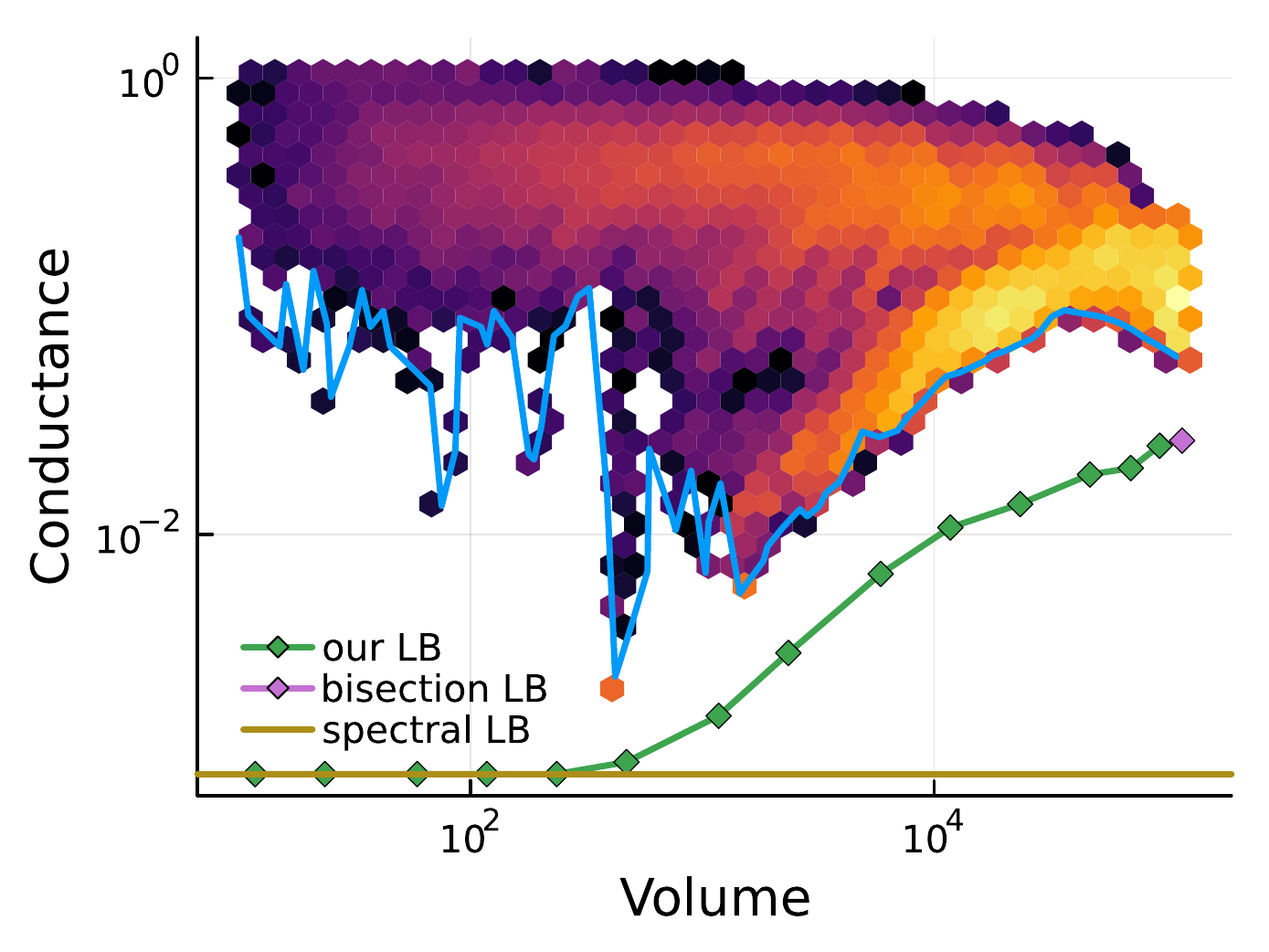}
    \caption{Comparison with spectral lower bounds (bottom yellow line) and minimum bisection SDP lower bounds (purple point at right) on two graphs Astro (Top) and HepPh (Bottom) graphs. We can observe that our \(\mu\)-
    conductance is capable of offering a lower bound at more positions and provides the expected smooth interpolation between these bounds that had been missing from existing approaches.}
    \label{fig:lowerbound_comp}
     \vspace{-\baselineskip}
\end{tuftefigure}

\subsection{Impact of rank on the lower bound.}
The rank parameter \(k\) plays a key role in our solution. As is shown
in Section~\ref{app:exp-running-time}, a higher \(k\) will slow down the computation. It also impacts the a posterori bound we achieve.  We 
 study this tradeoff here. The results on AstroPh and 
HepPh graphs are shown in Figure~\ref{fig:different_k}. We do not observe a strong pattern.   Consequently, we recommend setting $k=5$ as a pragmatic middle ground. 

\begin{tuftefigure}[h!]
\centering
        \includegraphics[width=0.45\linewidth]{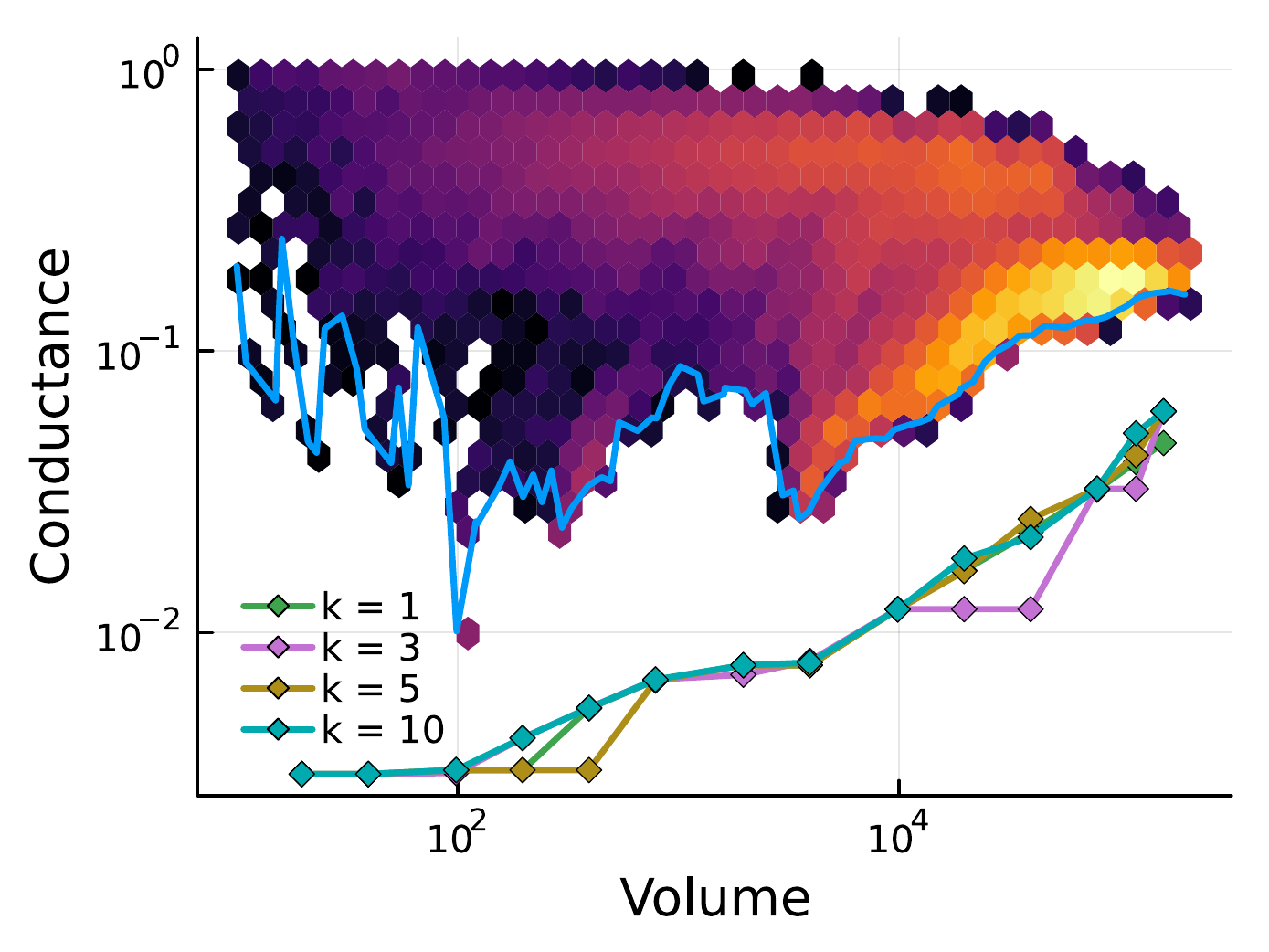} 
        \includegraphics[width=0.45\linewidth]{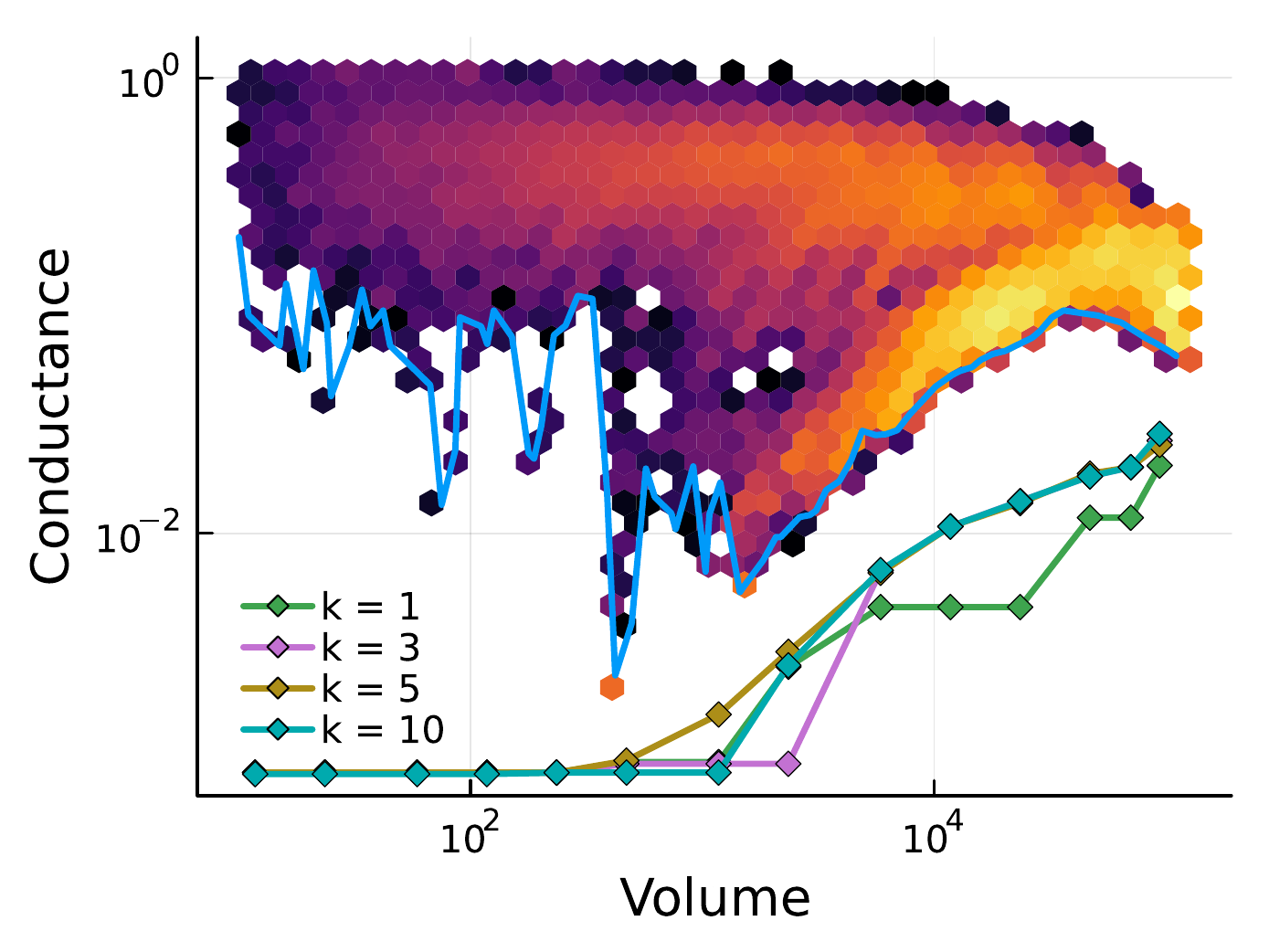}     
        \caption{The effect of rank parameter \(k\) on the lower bound illustrated on Astro (Top) and HepPh (Bottom). We observe that different
    \(k\) exhibit comparable curves but extremely small \(k\) will get worse lower bounds.}
    \label{fig:different_k}
\end{tuftefigure}


\section{Discussion and Future work}
The theorem and algorithm here allows a complete characterization of the network community profile for graphs with over 200,000 vertices. This, in turn, has implications on random sampling on real-world graphs as discussed in the introduction.
There were some theoretical datapoints known regarding bounds on the NCP. For instance, the position of the spectral partitioning set and the associated Cheeger inequality gives one point in the size-vs-conductance space. Another point arises from SDP-based methods~\cite{Leskovec-2009-community-structure} (Section 5.2) for bisection splits. Our tools are the first to interpolate between the two with robust bounds. Our code is available at \url{https://github.com/luotuoqingshan/mu-conductance-low-rank-sdp}. 

Our methods involve choosing a rank parameter, a value of $\mu$, as well as tolerances associated with the L-BFGS-B based procedure. These can have non-trivial interactions. 
Typically, we find that the values of $\theta$ involved in the lower bound are small (think $10^{-4})$. 
In a counter-intuitive observation, we found instances where using a weaker or higher tolerance values results in better or larger lower bounds on the $\mu$-conductance value because the value of $\theta$ was changed. In other scenarios, we found values of $\mu$ where we could not find a way to adjust rank and tolerance to make the value of $\theta$ small enough. This made the lower bound was extremely loose (or even negative in some scenarios). Our choice of overall parameters tends to minimize this. 

Although we have focused on lower bounds in this manuscript (in the interest of space). In a related line of work, we have developed a two-sided bound on the $\mu$-conductance spectral program \eqref{eq:mu_spectral_cut}~\cite{huang2023cheeger}. This gives a full Cheeger-like characterization of this program. There are also numerous variants of the Cheeger inequality \cite{louis2012many, kwok2013improved, koutis2014generalized, zhu2015parallel},
including those versions using multiple eigenvalues as well as more general weightings. Finding a multivector and multiset generalization of these results would be useful in a variety of scenarios. 

At the moment, we are able to handle a variety of real-world graphs, but the runtime is still slow. Computations take days instead of hours. Scaling these algorithms up to the LiveJournal example from Figure~\ref{fig:intro} is another challenge, with many potential avenues including parallelization. Delocalized eigenvectors for spectral clustering also arise from the statistics perspective in terms of regularization~\cite{Amini-2013, NEURIPS2018_2a845d4d}. Many of these techniques involve directly regularizing the graph Laplacian by adding a small multiple of the all ones matrix, akin to the PageRank perturbation. We hope to study relationships between these regularization techniques and $\mu$-conductance in the future, especially as they would help promise much faster runtimes via eigenvector techniques instead of low-rank SDPs.

\begin{fullwidth}
\bibliographystyle{dgleich-bib}
\bibliography{main.bib}
\end{fullwidth}

\appendix

\newpage

\end{document}